\documentclass[twocolumn]{autart}    
\usepackage{graphicx}
\usepackage{amsfonts,latexsym,
amssymb,
amsmath,url,stackengine}
\usepackage{booktabs}
\usepackage{amsmath}
\usepackage{picins}
\usepackage{color}
\usepackage{empheq}
\usepackage{mathrsfs}

\usepackage{rgsMacros}
\usepackage{subcaption}
\usepackage{bm}
\let\labelindent\relax
\usepackage{enumitem}
\usepackage{balance}
\usepackage{comment}

\definecolor{olivegreen}{rgb}{0,0.5,0}

\newtheorem{thmm}{Theorem}
\newtheorem{propo}{Proposition}
\newtheorem{lemm}{Lemma}
\newtheorem{exe}{Example}
\newtheorem{corol}{Corollary}
\newtheorem{ass}{Assumption}
\newtheorem{prope}{Property}
\newtheorem{defin}{Definition}
\newtheorem{remm}{Remark}

\newenvironment{lemma}{\begin{lemm}}{\hfill \end{lemm}}

\newenvironment{remark}{\begin{remm} \rm}{\hfill \end{remm}}

\newenvironment{assumption}{\begin{ass}}{\end{ass}}

\newenvironment{theorem}{\begin{thmm}}{\hfill   \end{thmm}}

\newenvironment{proposition}{\begin{propo}}{\hfill   \end{propo}}

\newenvironment{proof}{{\it Proof. }}{\hfill $\blacksquare$ }

\newtheorem{dwellt}{Condition}

\newtheorem{feedback}{Feedback}

\newtheorem{adapt_alg}{Algorithm}

\newif\ifitsdraft

\makeatletter
\def\@citex[#1]#2{%
  \let\@citea\@empty
  \@cite{\@for\@citeb:=#2\do
    {\@citea\def\@citea{,\penalty\@m\hskip.1em}%
     \edef\@citeb{\expandafter\@firstofone\@citeb\@empty}%
     \if@filesw\immediate\write\@auxout{\string\citation{\@citeb}}\fi
     \@ifundefined{b@\@citeb}{\mbox{\reset@font\bfseries ?}%
        \G@refundefinedtrue
        \@latex@warning
          {Citation `\@citeb' on page \thepage \space undefined}}%
        {\hbox{\csname b@\@citeb\endcsname}}}}{#1}}
\makeatother
 
\begin{document}

\begin{frontmatter}
\title{Region of Attraction of a Backstepping Observer\\ for the Quasilinear Heat Equation} 

\author[UGA]{M. C. Belhadjoudja}\ead{mohamed.belhadjoudja@gipsa-lab.fr}, 
\author[WU]{K. A. Morris}\ead{kmorris@uwaterloo.ca} 

\address[UGA]{Universit\'e Grenoble
Alpes, CNRS, Grenoble-INP, GIPSA-lab, F-38000, Grenoble, France}

\address[WU]{Department of Applied Mathematics, University of Waterloo, 200 University Avenue West, Waterloo, ON, Canada, N2L 3G1}


\begin{abstract}
It has long been of interest to understand when control and observer-design methodologies developed for linear dynamical systems can provide local guarantees when applied to nonlinear systems. However, beyond establishing local stability, it is also important to characterize the region of attraction and performance in terms of system parameters and design gains—something that linearization-based approaches cannot achieve. For instance, although increasing an observer gain improves the convergence rate in the linear setting (as in backstepping design), the relationship between this gain and the convergence rate in the nonlinear setting need not be monotonic; understanding these dependencies is therefore crucial for reliable operation. In light of this discussion, we consider the one-dimensional quasilinear heat equation with state-dependent diffusivity, and design a boundary-output observer based on the backstepping design for a linear heat equation with constant diffusivity. Viewing the quasilinear system as a perturbation of the linear one, we establish exponential stability of the origin for the estimation error dynamics in $H^1$, with an estimate of the region of attraction depending on the system parameters, observer gains, and the mismatch between the nonlinear diffusivity and the constant design diffusivity. Importantly, the estimation error converges to zero rather than merely to a neighborhood scaling with this mismatch, even though, in contrast to backstepping-based stabilization of nonlinear PDEs, the mismatch need not decay along trajectories and may remain bounded away from zero, acting as a persistent state-dependent multiplicative perturbation. A technical challenge was to perform a sufficiently-fine Lyapunov analysis that does not yield overly conservative results such as mere boundedness of the estimation error.
\end{abstract}

\end{frontmatter}

\setlength{\abovedisplayskip}{4pt}
\setlength{\belowdisplayskip}{4pt}
\setlength{\abovedisplayshortskip}{2pt}
\setlength{\belowdisplayshortskip}{2pt}

\section{Introduction} \label{Introduction}

The quasilinear diffusion equation models a number of important situations. For example, it can be used to model heat transfer in high-temperature plasmas via radiation \cite[Pages 652--684]{plasma}. Another important class of applications is  heat diffusion in phase-change materials (PCM)s, such as paraffin or fatty acids, where the heat capacity varies sharply near the transition temperature. Other contexts include the modeling of gas density in a porous medium \cite{porous1}, the filtration of an incompressible fluid through a porous stratum \cite{boussinesq}, and the density of biological species in a habitat \cite[Section 2.4]{vazquez}. 
Additional applications, such as coagulation-fragmentation processes, magma propagation, and the evolution of immiscible fluids, are discussed in \cite[Chapter 21]{vazquez}.

It is generally impossible to measure the state at every spatial location. Consequently, for tasks such as monitoring and feedback control, one must rely on measured outputs. This paper focuses on the problem of observer design for  one-dimensional diffusion using measurement of the state at the boundary.

In \cite{camil_finite_2}, stabilizing boundary control under state-dependent diffusivity was considered. However, the methodology proposed in the aforementioned work does not permit tuning of the convergence rate of the state towards the equilibrium. Consequently, even if this methodology was extended to solve the observer-design problem, the resulting convergence rate would not be adjustable. On the other hand, \cite{hugo} designed a finite-dimensional output-feedback boundary controller. Their method  provides stability guarantees in the $H^{-1}$ norm. This limitation makes it difficult to handle practically relevant cases in which diffusivity does not satisfy a linear growth condition. In \cite{quasi_1}, observer design with periodic boundary conditions was investigated via system linearization. However, while linearization-based techniques can be effective for certain classes of PDEs, including parabolic ones \cite{kirsten3}, they generally do not provide information on the region of attraction. More fundamentally, linearization erases the structure that governs how performance degrades as design parameters vary, so that even qualitative questions --- such as whether increasing an observer gain always improves convergence --- cannot be addressed within that framework. 

It is natural to ask whether observer-design techniques developed for the linear heat equation can be extended to the quasilinear case by treating the nonlinearity as a state-dependent perturbation. This perturbation approach has proven effective for the stabilization of one-dimensional quasilinear hyperbolic equations \cite{coron} and semilinear parabolic equations \cite{semi}. However, to the best of our knowledge, this approach has not been applied to quasilinear parabolic equations. Moreover, the aforementioned works did not analyze the region of attraction nor examine how performance depends on system parameters. There are  challenges that exist in extending linear observer design to a quasilinear system  that do not exist in the corresponding problem of stabilization. When applying a linear controller to a nonlinear system with the goal of stabilizing an equilibrium, there is usually a mismatch between the nonlinearity of the system and some of the controller gains. In the stabilization problem, as the state converges to zero, this mismatch also vanishes. For the observer-design problem, this is not the case, and one must carefully select how to handle the nonlinear terms. 

This article proposes a backstepping observer for the one-dimensional quasilinear heat equation subject to Neumann boundary conditions, with the output being the state at one end of the spatial interval. The observer consists of a copy of the quasilinear heat equation with both distributed and boundary injection terms. The gains for these terms are those obtained via the backstepping approach for a linear heat equation with constant diffusivity, as considered in \cite{back_linear}. We use these same gains in the nonlinear observer and derive conditions on the system parameters that guarantee exponential stability of the zero equilibrium for the estimation error dynamics in the $H^1$ norm. In this work, it is shown that the linear filter can be used to obtain an observer for the quasilinear system. Furthermore, a lower-bound on the region of attraction is obtained in terms of the observer gains and system parameters. The $H^1$ exponential stability result implies the exponential convergence of the max norm of the estimation error to zero, for all initial estimation errors in the region of attraction. Our results are obtained by analyzing the variations of a suitable Lyapunov functional candidate. One of the main technical difficulties is to appropriately handle the mismatch between the nonlinear diffusivity of the quasilinear heat equation and the constant observer gain that plays the role of the diffusivity in the linear case. Since we do not consider stabilization of the quasilinear heat equation, the mismatch due to the nonlinearity does not vanish along the system trajectories, unlike in backstepping control for the stabilization of other types of PDEs, such as quasilinear hyperbolic equations \cite{coron} and semilinear parabolic equations \cite{semi}. In our Lyapunov inequalities, the mismatch appears as a persistent, state-dependent multiplicative perturbation. Among the different ways of estimating the resulting terms, we identify bounds that avoid treating this mismatch as an external additive disturbance, which would lead to overly conservative input-to-state stability estimates.

A preliminary conference version of this work has been accepted for publication at the 2026 IEEE Conference on Decision and Control. The conference version does not include the proofs of our results and provides less extensive discussion and numerical simulations.

\section{Notation}

We denote by $\mathbb{R}_{>0}$ (resp., by $\mathbb{R}_{\geq 0}$) the set of positive (resp., of nonnegative) real numbers. Given a function $z \in \mathscr{C}^1([0,1];\mathbb{R})$, where $\mathscr{C}^1([0,1];\mathbb{R})$ denotes the space of continuously differentiable functions $z:[0,1]\to \mathbb{R}$, we let $|z|_{\infty} := \max_{x \in [0,1]}|z(x)|$, $|z|_{L^2}^2 := \int_{0}^1z(x)^2dx$, and $|z|_{H^1}^2 := |z|_{L^2}^2+|z'|_{L^2}^2$, where $z'$ denotes the derivative of $z$. Similarly, given a continuous function
$z : [0,1] \times [0,+\infty) \to \mathbb{R}$ and $t\geq 0$, we let $|z(\cdot,t)|_{\infty}:=\max_{x \in [0,1]}|z(x,t)|$, $|z(\cdot,t)|_{L^2}^2:=\int_{0}^{1}z(x,t)^2dx$, and $|z(\cdot,t)|_{H^1}^2 := |z(\cdot,t)|_{L^2}^2+|z_x(\cdot,t|_{L^2}^2$, where $z_x$ denotes the partial derivative of $z$ with respect to $x$. We also write $|z|_{\infty}$ and $|z|_{L^2}$ to mean the functions $t \mapsto |z(\cdot,t)|_{\infty}$ and $|z(\cdot,t)|_{L^2}$, respectively.
Furthermore, we denote by $z_{xx}$ the second partial derivative of $z$ with respect to $x$, and by $z_t$ the partial derivative of $z$ with respect to $t$. Given $\beta \in (0,1)$, we denote by $\mathscr{H}^{2+\beta}([0,1])$ 
the space of twice continuously-differentiable functions $z: [0,1] \rightarrow \mathbb{R}$ such that $z''$, the second-order derivative of $z$, is H\"{o}lder continuous with exponent $\beta$, i.e., there exists $C>0$ such that 
$ |z''(x_1)-z''(x_2)| \leq C |x_1-x_2|^{\beta}$ for all $x_1,x_2 \in [0,1]$.  
We denote by $\mathscr{H}_{loc}^{2+\beta,1+\beta/2}([0,1]\times [0,+\infty))$ the space of functions $z :[0,1] \times [0,+\infty) \to \mathbb{R}$ such that the restriction of $z$ to any bounded subset $[0,1] \times L \subset [0,1]\times [0,+\infty)$ belongs to $\mathscr{H}^{2+\beta,1+\beta/2}([0,1] \times L)$, the space of functions $f: [0,1]\times L\to \mathbb{R}$ such that $f_{xx}$ and $f_t$
are H\"older continuous in $x$ with exponent $\beta$ uniformly in $t$
and H\"older continuous in $t$ with exponent $\beta/2$ uniformly
in $x$. Finally, given a function $z:[0,1]\times [0,+\infty)\to \mathbb{R}$, we let $\|z\|_{\infty} := \sup_{(x,t)\in [0,1]\times [0,+\infty)} |z(x,t)|$.

\section{System Description}\label{system_description}
This article addresses the problem of observer design for the quasilinear heat equation
\begin{empheq}[left=\Sigma: \left\{,right=\right.]{align}
z_t &= \left( \alpha(z) z_x \right)_x, \label{sigma_1_z}\\
z_x(0,t)  &= z_x(1,t) = 0, \label{sigma_21}\\
z(x,0) &= z_o(x), \label{sigma_3_z}
\end{empheq}
where $(x,t)\mapsto z(x,t) \in \mathbb{R}$ is the state with $x\in [0,1]$ the space variable and $t\geq 0$ the time variable, $x\mapsto z_o(x) \in \mathbb{R}$ is the initial condition, and $\alpha : \mathbb{R}\to \mathbb{R}_{>0}$ is the diffusivity. 

\begin{defin}\label{def_classical}
Let $T\in (0,+\infty) \cup \{+\infty\}$. A function $z: [0,1]\times [0,T)\to \mathbb{R}$ is said to be a classical solution to $\Sigma$ on $[0,T)$ if: 
\begin{itemize}
    \item $z\in \mathscr{H}_{loc}^{2+\beta, 1+\beta/2}([0,1]\times [0,T))$ for some $\beta \in (0,1)$.
    \item \eqref{sigma_1_z} holds for all $(x,t) \in (0,1)\times (0,T)$. 
    \item \eqref{sigma_21} holds for all $t\in [0,T)$. 
    \item \eqref{sigma_3_z} holds for all $x\in [0,1]$.
\end{itemize}
A forward-complete classical solution to $\Sigma$ is any classical solution to $\Sigma$ on $[0,+\infty)$.
\hfill $\bullet$
\end{defin}

The solutions to all evolution PDEs considered in this article are defined analogously to Definition \ref{def_classical}.

To guarantee the existence and uniqueness of a forward-complete classical solution to $\Sigma$, we impose the following regularity and compatibility assumption.

\begin{assumption}\label{ass_1}
The following conditions hold.
\begin{itemize}
    \item There exists $\beta \in (0,1)$ such that
    $z_o\in \mathscr{H}^{2+\beta}([0,1])$.
    \item $\alpha \in \mathscr{C}^2(\mathbb{R};\mathbb{R}_{>0})$.
    \item $z'_o(0)=z'_o(1)=0$.
\end{itemize}
\hfill $\bullet$
\end{assumption} 

\begin{proposition}[Page 491, Theorem 7.4 in \cite{ladyzen}]
    Let Assumption \ref{ass_1} hold. Then, $\Sigma$ admits a unique forward-complete classical solution. 
\end{proposition}

Let $z$ be the unique forward-complete classical solution to $\Sigma$. Given the regularity of $z$, we can consider the output
\begin{align}
y(t) := z(1,t) \qquad \forall t\geq 0. 
\end{align}

The objective is to design an observer that estimates the state $z$ from the output $y$, with the estimation error converging exponentially to zero in the $H^1$ norm and, consequently, in the max norm. We further seek to characterize a region of attraction and a convergence rate of the estimation error to zero. To establish these results, we require the following assumption.

\begin{assumption}\label{ass2}
We know constants $M, M', M''\geq 0$ such that 
\begin{align*}
\|z\|_{\infty}\leq M, \ \ \|z_x\|_{\infty}\leq M', \ \ \|z_{xx}\|_{\infty}\leq M''.  
\end{align*}
\hfill $\bullet$
\end{assumption}

The constants in Assumption \ref{ass2} are not required for designing the observer; they only enter the estimates of the region of attraction and the convergence rate. 

\begin{remark}\label{rem_1}
There is no loss of generality in considering \eqref{sigma_1_z} instead of
$$c(z)z_t = \left(k(z)z_x\right)_x,$$
for some functions $c,k\in \mathscr{C}^2(\mathbb{R};\mathbb{R}_{>0})$. Indeed, we can define the enthalpy transformation \cite[Section 8.8.3]{rou}
\begin{align*}
\hat{c}(z) := \int_0^z c(\rho) d\rho \qquad \forall z \in \mathbb{R},
\end{align*}
and introduce the new state variable
\begin{align*}
\xi := \hat{c}(z).
\end{align*}
Since $c>0$, $\hat{c}$ is strictly increasing and hence invertible. One readily verifies that
$$\xi_t = \left(\alpha(\xi)\xi_x\right)_x,$$
where $\alpha := (k \circ \hat{c}^{-1})/(c \circ \hat{c}^{-1}) \in \mathscr{C}^2(\mathbb{R};\mathbb{R}_{>0})$. Furthermore, $\xi_x(0,t)=\xi_x(1,t)=0$, while the output becomes $\xi(1,t) = \hat{c}(z(1,t))$, which is a nonlinear transformation of the original output $z(1,t)$. 

We can apply to the PDE governing $\xi$ our observer below and obtain $H^1$ exponential stability of the zero equilibrium for the estimation error dynamics. Returning to the original variables, we then obtain estimates on the $H^1$ norm of the estimation error $z-\hat{c}^{-1}(\hat{\xi})$, where $\hat{\xi}$ is the estimate of $\xi$. These estimates depend on the Lipschitz constant of $\hat{c}^{-1}$ on the bounded set in which $|z(\cdot,t)-\hat{c}^{-1}(\hat{\xi}(\cdot,t))|_{\infty}$ takes its values.
\hfill $\bullet$
\end{remark}

\section{Observer Design}\label{observer_design}

The general structure of the observer considered for $\Sigma$ is
\begin{equation*}
\hat{\Sigma} : \left\lbrace 
\begin{aligned}
\hat{z}_t &= \left(\alpha (\hat{z})\hat{z}_x\right)_x + p_1(x)\left(y(t)-\hat{y}(t)\right), \\
\hat{z}_x(0,t) &= 0, \\
\hat{z}_x(1,t) &= p_{10}\left(y(t)-\hat{y}(t)\right), \\
\hat{z}(x,0) &= \hat{z}_o(x),
\end{aligned}
\right.
\end{equation*}
where $\hat{z}:[0,1]\times [0,+\infty)\to \mathbb{R}$ is the observer's state, $\hat{y}(t):=\hat{z}(1,t)$ for all $t\geq 0$, $\hat{z}_o\in \mathscr{H}^{2+\beta}([0,1])$ is the observer's initial condition that is required to satisfy the compatibility condition 
\begin{align}
\hat{z}'_o(0) = 0, \quad \hat{z}'_o(1) = p_{10}\left(y(0)-\hat{y}(0)\right), \label{compat2}
\end{align}
necessary for the existence of a classical solution to $\hat{\Sigma}$, and $(p_1,p_{10})$ are observer gains to be designed. 

In the linear case, where $\alpha:=a$ for some constant $a>0$, the gains $(p_1,p_{10})$ can be designed via backstepping to achieve global $H^1$ exponential stability of the zero equilibrium for the estimation error dynamics \cite{back_linear}.

To design the gains in the linear case, we start by defining the Bessel functions
\begin{align}
I_1(s) &:= \sum_{m=0}^{+\infty} \frac{(s/2)^{1+2m}}{m!(m+1)!}, \label{bessel_I}\\
J_1(s) &:= \sum_{m=0}^{+\infty} \frac{(-1)^m(s/2)^{1+2m}}{m!(m+1)!}.\label{bessel_J}
\end{align}
Furthermore, we define the functions $p, l \in \mathscr{C}^2(\mathcal{T})$, where 
\begin{align}
\mathcal{T} := \{(x,y)\in [0,1]^2 : 0\leq x\leq y\leq 1\},
\end{align}
by 
\begin{equation}\label{kernel_1}
p(x,y) := 
\left\lbrace 
\begin{aligned}
&-\frac{\sigma}{a} y\frac{I_1\left(\sqrt{\frac{\sigma}{a} (y^2-x^2)}\right)}{\sqrt{\frac{\sigma}{a} (y^2-x^2)}} &&\text{if $y\neq x$}, \\
&-\frac{\sigma}{2a}y &&\text{if $y=x$},
\end{aligned}
\right. 
\end{equation}
\begin{equation}\label{kernel_inv}
l(x,y) := 
\left\lbrace 
\begin{aligned}
&-\frac{\sigma}{a} y\frac{J_1\left(\sqrt{\frac{\sigma}{a} (y^2-x^2)}\right)}{\sqrt{\frac{\sigma}{a} (y^2-x^2)}} &&\text{if $y\neq x$}, \\
&-\frac{\sigma}{2a}y &&\text{if $y=x$},
\end{aligned}
\right. 
\end{equation}
where $\sigma>0$ is a constant free design parameter that will tune the convergence rate of the observer. 

We refer to $p$ as the backstepping kernel, and $l$ as the inverse backstepping kernel. 

The gains of the observer are 
\begin{align}
p_1(x) &:= -a p_y(x,1) \quad \forall x\in (0,1), \label{p1}\\
p_{10} &:= -p(1,1). \label{p10}
\end{align}

In the nonlinear case, we propose to use for $\hat{\Sigma}$ the gains in \eqref{p1}-\eqref{p10} designed for the linear heat equation, but this time both $a$ and $\sigma$ are free design parameters. As the existence and uniqueness of a unique forward-complete classical solution to the observer follows from standard arguments \cite[Chapter V]{ladyzen}, the details are omitted.

Given a forward-complete classical solution to $\hat{\Sigma}$, we define the estimation error 
\begin{align}
e := z-\hat{z}. \label{est_err}
\end{align}
Note that, because of Assumption \ref{ass_1}, the regularity of $\hat{z}_o$ and \eqref{compat2}, $e_o := z_o-\hat{z}_o \in \mathscr{H}^{2+\beta}([0,1])$ and 
\begin{align}
e_o'(0)=0, \ \ e_o'(1)=-p_{10}e_o(1). \label{compat3}
\end{align}

To characterize the region of attraction towards the equilibrium $\{e=0\}$ for the estimation error dynamics, we introduce, for each $\omega>0$, the set 
\begin{align*}
\mathscr{E}_o(\omega) := \{e_o\in \mathcal{H}^{2+\beta}([0,1]): |e_o|_{H^1} \leq \omega \ \text{and \eqref{compat3} holds}\}. 
\end{align*}

\section{Main Results}\label{main_result}

The main contribution of this paper is to determine a constant $\omega^*>0$, depending only on the system parameters and observer gains $a$ and $\sigma$, such that the the origin $\{e=0\}$ for the estimation error dynamics is exponentially stable in the $H^1$ norm with $\mathscr{E}_o(\omega^*)$ being a region of attraction. Specifically, for a fixed $z_o$, we find $\omega^*>0$, $M^*\geq 1$ and $\sigma^*>0$, such that for all $\hat{z}_o$ for which $e_o\in \mathscr{E}_o(\omega^*)$,
\begin{align}
|e(\cdot,t)|_{H^1}\leq M^* |e_o|_{H^1} \exp\left(-\frac{\sigma^*}{2} t\right) \quad \forall t\geq 0.
\end{align}
In particular, by Agmon's inequality, the $H^1$ exponential stability result implies the exponential convergence of the max norm of the estimation error to zero, for all initial estimation errors $e_o\in \mathscr{E}_o(\omega^*)$. 

Since our goal is to obtain estimates of the region of attraction and the convergence rate, we need expressions for the various constants and functions appearing in the estimates. In this section, we collect these expressions for ease of reference and use them freely throughout the remainder of the paper.

We start by defining the \textit{diffusivity mismatch}
\begin{align}
\bar{\alpha}(z) := \alpha(z)-a \qquad \forall z\in \mathbb{R}. \label{mismatch_def}
\end{align}
In the linear case where $\bar{\alpha} = 0$, the estimation error dynamics reduce to a system that is globally exponentially stable at rate $\sigma$. That is, regardless of the magnitude of the initial estimation error $z_o-\hat{z}_o$, one has
\begin{align*}
|e(\cdot,t)|_{H^1}
\leq
C |e_o|_{H^1}
\exp\left(-\frac{\sigma}{2}t\right)
\quad \forall t\geq 0,
\end{align*}
where $C\geq 1$ is a constant depending only on certain norms of $p$ and $l$, and hence only on $a$ and $\sigma$. In our nonlinear setting, the nonzero $\bar{\alpha}$ acts as a \textit{multiplicative state-dependent perturbation} to that nominal system. The magnitude of this perturbation is measured by 
\begin{align}
\bar{\delta} := \max_{|z|\leq M} |\alpha(z)-a|. \label{delta_bar}
\end{align}

Furthermore, we let $L_{\alpha}(s)$ (resp., $L_{\alpha'}(s)$), for some $s\geq 0$, be the Lipschitz constant of $\alpha$ (resp., of $\alpha'$) on $[-s,s]$. We recall that $\alpha$ and $\alpha'$ are Lipschitz on bounded sets because $\alpha$ is of class $\mathscr{C}^2$, and, for all $s\geq 0$,
\begin{align}
L_{\alpha}(s) &= \max_{\tau \in [-s,s]}\{|\alpha'(\tau)|\}, \label{lipschitz1} \\
L_{\alpha'}(s) &= \max_{\tau \in [-s,s]}\{|\alpha''(\tau)|\}.\label{lipschitz2}
\end{align}

Next, we introduce the transformed state 
\begin{align}
w(x,t) := e(x,t) + \int_x^1l(x,y)e(y,t)dy, \label{inv_back}
\end{align}
satisfying 
\begin{align}
e(x,t) = w(x,t) - \int_x^1p(x,y)w(y,t)dy.\label{f_back}
\end{align}
We let 
\begin{align}
V(w(\cdot,t)) &:= \frac{1}{2}|w(\cdot,t)|_{L^2}^2, \\ 
H(w(\cdot,t)) &:= \frac{1}{2}|w_x(\cdot,t)|_{L^2}^2, \\
E(w(\cdot,t)) &:= V(w(\cdot,t))+H(w(\cdot,t)).
\end{align}

Next, we define 
\begin{align*}
|p|_{L^2}^2 &:= \int_{0}^{1}\int_{x}^{1}p(x,y)^2dydx, \\
|l_y(\cdot,1)|_{L^2}^2 &:= \int_{0}^{1}l_y(x,1)^2dx,
\end{align*}
and we use similar notations for the other kernel norms.

Finally, Table \ref{tab1} introduces a set of constants $\{\gamma_i\}_{i=1}^{15}$ and functions $\{\varepsilon_i\}_{i=1}^{8}$ that will be used throughout the paper. The functions $\{\varepsilon_i\}_{i=1}^{8}$ are nonnegative, continuous, and nondecreasing with respect to both $E$ and $\bar{\delta}$.

\begin{table}
\centering
\renewcommand{\arraystretch}{2}
\begin{tabular}{@{}ll@{}}
\toprule
$\gamma_1$ & $\begin{aligned}[t]
&3\left(1+\frac{\sigma}{2a}\right)\left(\left(\frac{\sigma}{2a}\right)^2+|p_x|_{L^2}^2\right)+\frac{\sigma}{2a}\\
&+\sup_{x\in (0,1)}\{|l_y(x,x)|\}\left(1+2\left(1+|p|_{L^2}^2\right)\right)\\
&+2\sqrt{2}|l_{yy}|_{L^2}\sqrt{1+|p|_{L^2}^2}
\end{aligned}$ \\
$\gamma_2$ & $\begin{aligned}[t]\left(\frac{\sigma}{2a}+\left(\frac{2a+\sigma}{a}\right)\left(1+|p|_{L^2}^2\right)\right)M'\end{aligned}$ \\
$\gamma_3$ & $\begin{aligned}[t]\frac{2\sigma}{a}+3|l_y(\cdot,1)|_{L^2}+\frac{3\sigma}{2a}|l(\cdot,1)|_{L^2}\end{aligned}$ \\
$\gamma_4$ & $\begin{aligned}[t](8a+3\sigma)/(4a)\end{aligned}$ \\
$\gamma_5$ & $\begin{aligned}[t]M'/2\end{aligned}$ \\
$\gamma_6$ & $\begin{aligned}[t]\frac{\sigma}{2a}+|l_y(\cdot,1)|_{L^2}+\frac{9\sigma}{2a}|l(\cdot,1)|_{L^2}\end{aligned}$ \\
$\gamma_7$ & $\begin{aligned}[t](8/a)\left(1+|l|_{L^2}^2\right)\end{aligned}$ \\
$\gamma_8$ & $\begin{aligned}[t](2\sigma^2/(a^3)\left(1+|l|_{L^2}^2\right)\end{aligned}$ \\
$\gamma_9$ & $\begin{aligned}[t]\frac{1}{a}\left(\frac{\sigma^2}{a^2}+24(M')^2\right)\left(1+|l|_{L^2}^2\right)\end{aligned}$ \\
$\gamma_{10}$ & $\begin{aligned}[t]\frac{2}{a}\left(\frac{1}{4}\left(\frac{\sigma}{a}\right)^4+4|p_x|_{L^4}^4\right)\left(1+|l|_{L^2}^2\right)\end{aligned}$ \\
$\gamma_{11}$ & $\begin{aligned}\frac{32}{a}\left(1+|l|_{L^2}^2\right)\end{aligned}$ \\
$\gamma_{12}$ & $\begin{aligned}\frac{16}{a}\left(1+|p|_{L^2}^2\right)\left(1+|l|_{L^2}^2\right)(M'')^2\end{aligned}$ \\
$\gamma_{13}$ & $\begin{aligned}[t]\frac{32}{a}\left(1+|p|_{L^2}^2\right)\left(1+|l|_{L^2}^2\right)(M')^4\end{aligned}$ \\
$\gamma_{14}$ & $\begin{aligned}[t]
&\frac{8}{a}\bigg(|p_{xx}|_{L^2}^2+\sup_{x\in (0,1)}\left\{\left(p_x(x,x)-\frac{\sigma}{2a}\right)^2\right\}\\
&+3\left(\frac{\sigma^2}{2a^2}+2|p_x|_{L^2}^2\right)(M')^2\bigg)\left(1+|l|_{L^2}^2\right)
\end{aligned}$ \\
$\gamma_{15}$ & $\begin{aligned}[t]
&\frac{16}{a}\bigg(|p_{xx}|_{L^2}^2+\sup_{x\in (0,1)}\left\{\left(p_x(x,x)-\frac{\sigma}{2a}\right)^2\right\}\bigg)\left(1+|l|_{L^2}^2\right)
\end{aligned}$ \\
$\varepsilon_1$ & $\begin{aligned}[t]
&\left\{2\gamma_1\left(1+|p|_{\infty}\right)\sqrt{E}+\gamma_2\right\}L_{\alpha}\left(2\left(1+|p|_{\infty}\right)\sqrt{E}\right)\\
&+2\gamma_3L_{\alpha}\left(2\sqrt{E}\right)\sqrt{E}+\gamma_{13}L_{\alpha'}\left(2\left(1+|p|_{\infty}\right)\sqrt{E}\right)^2\\
&+\left\{\gamma_{12}+4\gamma_{15}\left(1+|p|_{\infty}\right)^2E\right\}L_{\alpha}\left(2\left(1+|p|_{\infty}\right)\sqrt{E}\right)^2
\end{aligned}$ \\
$\varepsilon_2$ & $\begin{aligned}[t]16\gamma_{10}L_{\alpha'}\left(2\left(1+|p|_{\infty}\right)\sqrt{E}\right)^2\left(1+|p|_{\infty}\right)^2E\end{aligned}$ \\
$\varepsilon_3$ & $\begin{aligned}[t]
&\left\{2\gamma_4\left(1+|p|_{\infty}\right)\sqrt{E}+\gamma_5\right\}L_{\alpha}\left(2\left(1+|p|_{\infty}\right)\sqrt{E}\right)\\
&+4\gamma_8L_{\alpha}\left(2\left(1+|p|_{\infty}\right)\sqrt{E}\right)^2\left(1+|p|_{\infty}\right)^2E\\
&+2\gamma_6L_{\alpha}\left(2\sqrt{E}\right)\sqrt{E}
\end{aligned}$ \\
$\varepsilon_4$ & $\begin{aligned}[t]4\gamma_{11}L_{\alpha'}\left(2\left(1+|p|_{\infty}\right)\sqrt{E}\right)^2\left(1+|p|_{\infty}\right)^2E\end{aligned}$ \\
$\varepsilon_5$ & $\begin{aligned}[t]4\gamma_{7}L_{\alpha}\left(2\left(1+|p|_{\infty}\right)\sqrt{E}\right)^2\left(1+|p|_{\infty}\right)^2E\end{aligned}$ \\
$\varepsilon_6$ & $\begin{aligned}[t]
&\varepsilon_1 + (\gamma_1+\gamma_3)\bar{\delta}+\gamma_{14}\bar{\delta}^2+ 8\left(\varepsilon_2+\gamma_{10}\bar{\delta}^2\right)E
\end{aligned}$ \\
$\varepsilon_7$ & $\begin{aligned}[t]
&\varepsilon_3+(\gamma_4+\gamma_6) \bar{\delta}+\gamma_9\bar{\delta}^2+4\left(\varepsilon_4+\frac{\gamma_{11}}{4}\bar{\delta}^2\right)E
\end{aligned}$ \\
$\varepsilon_8$ & $\begin{aligned}[t]\varepsilon_5+\frac{\gamma_7}{2}\bar{\delta}^2 + 2\left(\varepsilon_4+\frac{\gamma_{11}}{4}\bar{\delta}^2\right)E\end{aligned}$ \\
\bottomrule \\
\end{tabular}
\caption{Constants and functions involved in the stability analysis.}
\label{tab1}
\end{table}

The first step of the stability analysis consists in deriving equations satisfied by the error $e$ and the transformed state $w$. For $e$, we just need to subtract $\hat{\Sigma}$ from $\Sigma$ to obtain 
\begin{equation*}
\Sigma_{e} : \left\lbrace 
\begin{aligned}
e_t &= \left(\alpha(z)z_x\right)_x-\left(\alpha(\hat{z})\hat{z}_x\right)_x-p_1(x)e(1,t), \\
e_x(0,t) &= 0, \\
e_x(1,t) &= -p_{10}e(1,t), \\
e(x,0) &= e_o(x) = z_o(x)-\hat{z}_o(x).
\end{aligned}
\right.
\end{equation*}
In the following lemma, we derive an equation for $w$. 
\begin{lemma}
It holds that
\begin{equation*}
\Sigma_{w} : \left\lbrace 
\begin{aligned}
w_t &= a w_{xx} - \sigma w + f, \\
w_x(0,t) &= w_x(1,t) = 0, \\
w(x,0) &= w_o(x) := e_o(x)+\int_x^1l(x,y)e_o(y)dy,
\end{aligned}
\right.
\end{equation*}
where 
\begin{align}
f(x) :=&~ (\bar{\alpha}(z(x))z_x(x))_x-(\bar{\alpha}(\hat{z}(x))\hat{z}_x(x))_x\nonumber \\
&~+\int_x^1l(x,y)(\bar{\alpha}(z(y))z_y(y))_ydy \nonumber\\
&~-\int_x^1l(x,y)(\bar{\alpha}(\hat{z}(y))\hat{z}_y(y))_ydy. \label{f_def}
\end{align}
\hfill  
\end{lemma}
\begin{proof}
Note that 
\begin{align*}
w_x(0,t) = e_x(0,t)-l(0,0)e(0,t)+\int_0^1l_x(0,y)e(y,t)dy.
\end{align*}
By the definition of $l$ in \eqref{kernel_inv}, we have $l(0,0)=l_x(0,y)=0$. Hence, $w_x(0,t)=e_x(0,t)=0$.  

Similarly, 
\begin{align*}
w_x(1,t) &= e_x(1,t) - l(1,1)e(1,t) \\
&= -p_{10}e(1,t)-l(1,1)e(1,t).
\end{align*}
By \eqref{p10} and \eqref{kernel_1}-\eqref{kernel_inv}, $l(1,1)=p(1,1)=-p_{10}$. Hence, $w_x(1,t)=0$. 

Next, we define
\begin{align*}
f(x) := w_t(x) - a w_{xx}(x)+\sigma w(x),
\end{align*}
and show that it is given by \eqref{f_def}. Differentiating both sides of \eqref{f_back} with respect to $t$, we obtain 
\begin{align*}
e_t&= w_t-\int_{x}^1p(x,y)w_t(y)dy \nonumber \\
&= aw_{xx}-\sigma w + f -a\int_{x}^1p(x,y)w_{yy}(y)dy \\
&~+ \sigma \int_{x}^1p(x,y)w(y)dy-\int_{x}^{1}p(x,y)f(y)dy. 
\end{align*}
Using integration by parts and the boundary condition $w_x(1)=0$, note that we have 
\begin{align*}
\int_{x}^1p(x,y)w_{yy}(y)dy =&~ \left[p(x,y)w_y(y)\right]_{y=x}^{y=1} \nonumber \\&~-\int_{x}^{1}p_y(x,y)w_y(y)dy \nonumber \\
=&~p(x,1)w_x(1)-p(x,x)w_x(x)  \nonumber \\
&~- \left[p_y(x,y)w(y) \right]_{y=x}^{y=1} \nonumber \\
&~+\int_{x}^1p_{yy}(x,y)w(y)dy \nonumber
\end{align*}
\begin{align*}
=&~ - p(x,x)w_x(x)-p_y(x,1)w(1) \nonumber \\
&~+p_y(x,x)w(x) \nonumber \\
&~+\int_{x}^1p_{yy}(x,y)w(y)dy. 
\end{align*}
As a result, we have
\begin{align*}
e_t =&~ aw_{xx}-\big(\sigma+ap_y(x,x)\big) w+f \nonumber \\
&~+ap(x,x)w_x+ap_y(x,1)w(1) \nonumber \\
&~+\int_{x}^1w(y) \big(\sigma p(x,y)-ap_{yy}(x,y)\big)dy\nonumber \\
&~-\int_{x}^1p(x,y)f(y)dy. 
\end{align*}
Next, we differentiate both sides of \eqref{f_back} with respect to $x$, to obtain 
\begin{align*}
e_x = w_x+p(x,x)w-\int_{x}^{1}p_x(x,y)w(y)dy.
\end{align*}
Differentiating a second time with respect to $x$, we find 
\begin{align*}
e_{xx} =&~w_{xx}+\left( \frac{d}{dx}p(x,x)\right)w+p(x,x)w_x \\
&~+p_x(x,x)w-\int_{x}^{1}p_{xx}(x,y)w(y)dy. 
\end{align*}
As a result, we have 
\begin{align*}
e_t-ae_{xx} +p_1(x)e(1)&= -\left(\sigma+2a\frac{d}{dx}p(x,x)\right)w\nonumber \\
&~+\left(ap_y(x,1)+p_1(x)\right)w(1) +f \nonumber \\
&~-\int_{x}^{1}p(x,y)f(y)dy \nonumber \\
&~+\int_x^1w(y)\big(\sigma p(x,y)-ap_{yy}(x,y) \nonumber \\
&~+ap_{xx}(x,y)\big)dy.
\end{align*}
By definition of $p$ in \eqref{kernel_1} and the choice of $p_1(x)$ in \eqref{p1}, we conclude that 
\begin{align*}
e_t-ae_{xx} +p_1(x)e(1) = f -\int_{x}^1p(x,y)f(y)dy. 
\end{align*}
Hence, using $\Sigma_e$, we obtain 
\begin{align*}
(\bar{\alpha}(z(x))z_x(x))_x-(\bar{\alpha}(\hat{z}(x))\hat{z}_x(x))_x &= f(x) \\
&-\int_{x}^{1}p(x,y)f(y)dy.
\end{align*}
Finally, \eqref{f_def} follows by applying the inverse backstepping transformation.
\end{proof}

The next step is to analyze $\dot{V}$, the time derivative of $t\mapsto V(w(\cdot,t))$. To this end, we first derive several auxiliary upperbounds that will enter into the estimate of $\dot{V}$. The following lemmas establish these bounds.

\begin{lemma}
It holds that 
\begin{align}
\bigg|\int_{0}^{1}&w\big(\left(\bar{\alpha}(z)-\bar{\alpha}(\hat{z})\right)(z_x-e_x)\big)_x dx\bigg| \nonumber 
\end{align}
\begin{align}
\leq&~ \bigg\{ \frac{3}{2}\left(\frac{\sigma^2}{2a^2}+2|p_x|_{L^2}^2\right)L_{\alpha}(|e|_{\infty})|e|_{\infty} \nonumber \\
&~+2M'\left(1+|p|_{L^2}^2\right)L_{\alpha}(|e|_{\infty}) \nonumber \\
&~+\frac{2\sigma}{a}L_{\alpha}(|e(1)|)|w(1)|\bigg\}V \nonumber \\
&~+\bigg\{2L_{\alpha}(|e|_{\infty})|e|_{\infty} +\frac{1}{2}M'L_{\alpha}(|e|_{\infty}) \nonumber \\
&~+\frac{\sigma}{2a}L_{\alpha}(|e(1)|)|w(1)|\bigg\}|w_x|_{L^2}^2. \label{F_B_1}
\end{align}
\hfill  
\end{lemma}
\begin{proof}
First note that 
\begin{align}
\int_{0}^{1}&w(x)\big(\left(\bar{\alpha}(z(x))-\bar{\alpha}(\hat{z}(x))\right)(z_x(x)-e_x(x))\big)_x dx \nonumber \\
=&~ \bigg[w(x)\left(\bar{\alpha}(z(x))-\bar{\alpha}(\hat{z}(x))\right)(z_x(x)-e_x(x))\bigg]_{x=0}^{x=1} \nonumber\\
&~-\int_{0}^{1}w_x(x)\left(\bar{\alpha}(z(x))-\bar{\alpha}(\hat{z}(x))\right)(z_x(x)-e_x(x))dx \nonumber 
\end{align}
\begin{align}
=&~\left(\bar{\alpha}(z(1))-\bar{\alpha}(\hat{z}(1))\right)p_{10}w(1)^2 \nonumber\\
&~-\int_{0}^{1}w_x(x)\left(\bar{\alpha}(z(x))-\bar{\alpha}(\hat{z}(x))\right)z_x(x)dx \nonumber \\
&~+\int_{0}^{1}w_x(x)\left(\bar{\alpha}(z(x))-\bar{\alpha}(\hat{z}(x))\right)e_x(x)dx. \label{first_1}
\end{align}
Furthermore, note that we have
\begin{align*}
|\bar{\alpha}(z(x))-\bar{\alpha}(\hat{z}(x))| &= |\alpha(z(x))-\alpha(\hat{z}(x))|\\
&\leq L_{\alpha}(|e(x)|)|e(x)| \\
&\leq L_{\alpha}(|e|_{\infty})|e(x)|. 
\end{align*}
Hence, applying Young's inequality, we obtain
\begin{align*}
&\left|\int_{0}^{1}w_x(x)\left(\bar{\alpha}(z(x))-\bar{\alpha}(\hat{z}(x))\right)z_x(x)dx\right| \\
&~\quad \leq \int_{0}^{1}\left|w_x(x)\left(\bar{\alpha}(z(x))-\bar{\alpha}(\hat{z}(x))\right)z_x(x)\right|dx \\
&~\quad \leq |z_x|_{\infty}L_{\alpha}(|e|_{\infty})\int_{0}^{1}|w_x(x)| |e(x)|dx \\
&~\quad \leq \frac{1}{2}M'L_{\alpha}(|e|_{\infty})|w_x|_{L^2}^2 + \frac{1}{2}M'L_{\alpha}(|e|_{\infty})|e|_{L^2}^2
\end{align*}
Moreover, squaring both sides of \eqref{f_back}, integrating with respect to $x$, and applying Young's and Cauchy-Schwarz inequalities, we have 
\begin{align}
|e|_{L^2}^2 =&~ \int_{0}^{1}\left(w(x) - \int_x^1 p(x,y)w(y)dy\right)^2dx \nonumber \\ 
=&~ |w|_{L^2}^2+\int_{0}^{1}\left(\int_x^1 p(x,y)w(y)dy\right)^2dx \nonumber\\
&~-2\int_{0}^{1}w(x)\left(\int_x^1 p(x,y)w(y)dy\right)dx\nonumber 
\end{align}
\begin{align}
\leq&~ 2|w|_{L^2}^2+2\int_{0}^{1}\left(\int_x^1 p(x,y)w(y)dy\right)^2dx \nonumber\\
\leq&~ 2|w|_{L^2}^2 + 2\int_{0}^{1}\left(\int_{x}^{1}p(x,y)^2dy \int_{x}^{1}w(y)^2dy\right) dx \nonumber\\
\leq&~ 4\left(1+|p|_{L^2}^2\right)V, \label{v_L2}
\end{align}
Consequently, 
\begin{align}
&\left|\int_{0}^{1}w_x(x)\left(\bar{\alpha}(z(x))-\bar{\alpha}(\hat{z}(x))\right)z_x(x)dx\right|\nonumber\\
&~\leq \frac{1}{2}M'L_{\alpha}(|e|_{\infty})|w_x|_{L^2}^2  + 2M'\left(1+|p|_{L^2}^2\right)L_{\alpha}(|e|_{\infty})V.\label{first_2}
\end{align}
Next, note that we have 
\begin{align*}
&\left|\int_{0}^{1}w_x(x)\left(\bar{\alpha}(z(x))-\bar{\alpha}(\hat{z}(x))\right)e_x(x)dx\right| \\
&~\quad \leq L_{\alpha}(|e|_{\infty})|e|_{\infty}\int_{0}^{1}|w_x(x)e_x(x)|dx \\
&~\quad \leq \frac{1}{2}L_{\alpha}(|e|_{\infty})|e|_{\infty}\left(|w_x|_{L^2}^2+|e_x|_{L^2}^2\right).
\end{align*}
Using \eqref{f_back} and the fact that $p(x,x)=-(\sigma/2a)x$, we obtain
\begin{align}
|e_x|_{L^2}^2 =&~ \int_{0}^{1}\bigg(w_x(x) +p(x,x)w(x) \nonumber \\
&~-\int_{x}^1p_x(x,y)w(y)dy\bigg)^2dx\nonumber\\
\leq&~ 3|w_x|_{L^2}^2+\frac{3\sigma^2}{4a^2}|w|_{L^2}^2 \nonumber\\
&~+3\int_{0}^{1}\left(\int_{x}^1p_x(x,y)w(y)dy\right)^2dx \nonumber \\
\leq&~3|w_x|_{L^2}^2+3\left(\frac{\sigma^2}{4a^2}+|p_x|_{L^2}^2\right)|w|_{L^2}^2 \nonumber\\
\leq&~3|w_x|_{L^2}^2+3\left(\frac{\sigma^2}{2a^2}+2|p_x|_{L^2}^2\right)V.\label{vx_L2_B}
\end{align}
Hence, we have 
\begin{align}
&\left|\int_{0}^{1}w_x(x)\left(\bar{\alpha}(z(x))-\bar{\alpha}(\hat{z}(x))\right)e_x(x)dx\right| \nonumber\\
&~\quad \leq 2L_{\alpha}(|e|_{\infty})|e|_{\infty}|w_x|_{L^2}^2 \nonumber\\
&~\quad \quad + \frac{3}{2}L_{\alpha}(|e|_{\infty})|e|_{\infty}\left(\frac{\sigma^2}{2a^2}+2|p_x|_{L^2}^2\right)V.\label{first_3}
\end{align}
Next, using the fact that $p_{10}=(\sigma/2a)$, we obtain 
\begin{align*}
\left(\bar{\alpha}(z(1))-\bar{\alpha}(\hat{z}(1))\right)p_{10}w(1)^2 \leq \frac{\sigma}{2a}L_{\alpha}(|e(1)|)|w(1)|w(1)^2.
\end{align*}
Applying Agmon's inequality, note that we have 
\begin{align*}
w(1)^2 \leq |w|_{L^2}^2+2|w|_{L^2}|w_x|_{L^2}. 
\end{align*}
As a result, 
\begin{align}
\left(\bar{\alpha}(z(1))-\bar{\alpha}(\hat{z}(1))\right)p_{10}&w(1)^2 \leq \frac{2\sigma}{a}L_{\alpha}(|e(1)|)|w(1)| V \nonumber \\
&~+\frac{\sigma}{2a}L_{\alpha}(|e(1)|)|w(1)| |w_x|_{L^2}. \label{first_f}
\end{align}
Combining \eqref{first_1}, \eqref{first_2}, \eqref{first_3}, and \eqref{first_f}, we obtain \eqref{F_B_1}. 
\end{proof}

\begin{lemma}
It holds that 
\begin{align}
\bigg|&\int_{0}^{1}w(x)\left( \bar{\alpha}(z(x))e_x(x)\right)_xdx\bigg| \nonumber \\
\leq&~ \bigg\{\frac{2\sigma }{a}+\frac{3}{2} \left(\frac{\sigma^2}{2a^2}+2|p_x|_{L^2}^2\right)\bigg\}\bar{\delta}V+\bigg\{\frac{\sigma }{2a}+2\bigg\}\bar{\delta}|w_x|_{L^2}^2. \label{F_B_2}
\end{align}
\hfill  
\end{lemma}
\begin{proof}
Using integration by parts, leads to
\begin{align*}
\int_{0}^{1}w(x)\left( \bar{\alpha}(z(x))e_x(x)\right)_x&dx = \left[w(x)\bar{\alpha}(z(x))e_x(x)\right]_{x=0}^{x=1} \nonumber \\
&~-\int_{0}^{1}w_x(x)\bar{\alpha}(z(x))e_x(x)dx \nonumber \\
=&~-\frac{\sigma}{2a}\bar{\alpha}(z(1))w(1)^2\nonumber \\
&~-\int_{0}^{1}w_x(x)\bar{\alpha}(z(x))e_x(x)dx. 
\end{align*}
By Agmon's inequality,
\begin{align*}
\frac{\sigma}{2a}\bar{\alpha}(z(1))w(1)^2 \leq 2\sigma \frac{|\bar{\alpha}(z(1))|}{a}V + \frac{\sigma |\bar{\alpha}(z(1))|}{2a}|w_x|_{L^2}^2. 
\end{align*}
Moreover, by Young's inequality and \eqref{vx_L2_B}, we obtain 
\begin{align}
\bigg|\int_{0}^{1}w_x(x)\bar{\alpha}(z(x))e_x(x)&dx\bigg|\leq \frac{|\bar{\alpha}(z)|_{\infty}}{2}\left(|w_x|_{L^2}^2+|e_x|_{L^2}^2\right) \nonumber \\
&\leq \frac{3}{2} \left(\frac{\sigma^2}{2a^2}+2|p_x|_{L^2}^2\right)|\bar{\alpha}(z)|_{\infty}V \nonumber \\
&~+2|\bar{\alpha}(z)|_{\infty}|w_x|_{L^2}^2. \nonumber 
\end{align}
As a result, we obtain \eqref{F_B_2}. 
\end{proof}

\begin{lemma}
It holds that 
\begin{align}
\bigg|&\int_{0}^{1}w(x)l(x,x)\bar{\alpha}(z(x))e_x(x)dx\bigg| \nonumber \\
&\leq \frac{\sigma}{2a}\left\{1+3\left(\left(\frac{\sigma}{2a}\right)^2+|p_x|_{L^2}^2\right)\right\}\bar{\delta}V+\frac{3\sigma}{4a}\bar{\delta}|w_x|_{L^2}^2. \label{F_B_3}
\end{align}
\hfill  
\end{lemma}
\begin{proof}
Using $l(x,x)=p(x,x)=-(\sigma/(2a))x$, we have 
\begin{align*}
\bigg|\int_{0}^{1}&w(x)l(x,x)\bar{\alpha}(z(x))e_x(x)dx\bigg| \\
&~\leq \frac{\sigma}{2a}|\bar{\alpha}(z)|_{\infty}\int_{0}^{1}|w(x)| |e_x(x)|dx \\
&~\leq \frac{\sigma}{2a}|\bar{\alpha}(z)|_{\infty} \left(V+\frac{1}{2}|e_x|_{L^2}^2\right).
\end{align*}
Using \eqref{vx_L2_B}, we obtain \eqref{F_B_3}. 
\end{proof}

\begin{lemma}
It holds that 
\begin{align}
\bigg|\int_{0}^{1}&w(x)l(x,x)\left(\bar{\alpha}(z(x))-\bar{\alpha}(\hat{z}(x))\right)z_x(x)dx \bigg| \nonumber \\
\leq&~ \frac{\sigma}{2a}M'L_{\alpha}(|e|_{\infty})\left(1+2\left(1+|p|_{L^2}^2\right)\right)V. \label{F_B_4}
\end{align}
\begin{align}
\bigg|\int_{0}^{1}&w(x)l(x,x)\left(\bar{\alpha}(z(x))-\bar{\alpha}(\hat{z}(x))\right)e_x(x)dx \bigg| \nonumber \\
\leq&~ \frac{\sigma}{2a}\left\{1+3\left(\left(\frac{\sigma}{2a}\right)^2+|p_x|_{L^2}^2\right)\right\}L_{\alpha}(|e|_{\infty})|e|_{\infty} V \nonumber \\
&~+\frac{3\sigma}{4a}L_{\alpha}(|e|_{\infty})|e|_{\infty}|w_x|_{L^2}^2. \label{F_B_5}
\end{align}
\end{lemma}
\begin{proof} 
\begin{align*}
\bigg|\int_{0}^{1}&w(x)l(x,x)\left(\bar{\alpha}(z(x))-\bar{\alpha}(\hat{z}(x))\right)z_x(x)dx \bigg| \\
\leq&~\frac{\sigma}{2a}|z_x|_{\infty}L_{\alpha}(|e|_{\infty})\int_{0}^{1}|w(x)||e(x)|dx \\
\leq&~ \frac{\sigma}{2a}M'L_{\alpha}(|e|_{\infty})\left( V+\frac{1}{2}|e|_{L^2}^2\right). 
\end{align*}
Then, using \eqref{v_L2}, we obtain \eqref{F_B_4}. The proof of \eqref{F_B_5} is similar. 
\end{proof}

\begin{lemma}
It holds that 
\begin{align}
&\bigg| \int_{0}^{1}w(x)l_y(x,1)\left(\int_{\hat{z}(1)}^{z(1)}\bar{\alpha}(s)ds\right) dx\bigg| \nonumber \\
&\leq 3\left( L_{\alpha}(|e(1)|) |e(1)|+\bar{\delta}\right)|l_y(\cdot,1)|_{L^2}V \nonumber \\
&+\left( L_{\alpha}(|e(1)|) |e(1)|+\bar{\delta}\right)|l_y(\cdot,1)|_{L^2}|w_x|_{L^2}^2. \label{F_B_6}
\end{align}
\hfill  
\end{lemma}
\begin{proof}
Since $\alpha$ is Lipschitz continuous on bounded sets, then, for each $\min\{z(1),\hat{z}(1)\}\leq s\leq \max\{z(1),\hat{z}(1)\}$, we have
\begin{align*}
|\bar{\alpha}(z(1))-\bar{\alpha}(s)|&\leq L_{\alpha}(|z(1)-s|)|z(1)-s| \\
&\leq  L_{\alpha}(|e(1)|) |e(1)|
\end{align*}
Hence,
\begin{align*}
|\bar{\alpha}(s)|&=|\bar{\alpha}(s)-\bar{\alpha}(z(1))+\bar{\alpha}(z(1))| \\
&\leq |\bar{\alpha}(s)-\bar{\alpha}(z(1))|+|\bar{\alpha}(z(1))| \\
&\leq L_{\alpha}(|e(1)|) |e(1)| + |\bar{\alpha}(z(1))|.
\end{align*}
As a result, 
\begin{align}
\bigg|\int_{\hat{z}(1)}^{z(1)}&\bar{\alpha}(s)ds\bigg| \nonumber \\
\leq&~ \int_{\min\{\hat{z}(1),z(1)\}}^{\max\{\hat{z}(1),z(1)\}}|\bar{\alpha}(s)|ds \nonumber \\
\leq&~ \int_{\min\{\hat{z}(1),z(1)\}}^{\max\{\hat{z}(1),z(1)\}}\left(L_{\alpha}(|e(1)|) |e(1)| + |\bar{\alpha}(z(1))|\right)ds \nonumber \\
\leq&~ L_{\alpha}(|e(1)|) |e(1)|^2+|\bar{\alpha}(z(1))||e(1)|.\label{int_bound}
\end{align}
Consequently, 
\begin{align*}
\bigg| \int_{0}^{1}&w(x)l_y(x,1)\left(\int_{\hat{z}(1)}^{z(1)}\bar{\alpha}(s)ds\right) dx\bigg| \\
\leq&~ \left( L_{\alpha}(|e(1)|) |e(1)|^2+|\bar{\alpha}(z(1))||e(1)|\right) \\
&\times \int_{0}^{1}|l_y(x,1)||w(x)|dx \\
\leq&~ \left( L_{\alpha}(|e(1)|) |e(1)|+|\bar{\alpha}(z(1))|\right) |l_y(\cdot,1)|_{L^2} |e(1)| \sqrt{2V}, 
\end{align*}
Using the fact that $e(1)=w(1)$ and applying Agmon's inequality, we obtain 
\begin{align}
|e(1)|\sqrt{2V} =&~ |w(1)|\sqrt{2V} \nonumber \\
\leq&~ \sqrt{2V \left(2V+2\sqrt{2V} |w_x|_{L^2}\right)} \nonumber \\
\leq&~\sqrt{2V\left(4V+|w_x|_{L^2}^2\right)} \nonumber \\
\leq&~ \sqrt{9V^2+|w_x|_{L^2}^4} \nonumber\\
\leq&~ 3V+|w_x|_{L^2}^2. \label{vV}
\end{align}
Hence, we obtain \eqref{F_B_6}. 
\end{proof}

\begin{lemma}
It holds that 
\begin{align}
\bigg| \int_{0}^{1}&w(x)l_y(x,x)\left(\int_{\hat{z}(x)}^{z(x)}\bar{\alpha}(s)ds\right)dx\bigg| \nonumber\\
\leq&~ \sup_{x\in (0,1)}\{|l_y(x,x)|\}\left(1+2\left(1+|p|_{L^2}^2\right)\right) \nonumber \\
&~\times \left(L_{\alpha}(|e|_{\infty}) |e|_{\infty}+\bar{\delta}\right)V. \label{F_B_7}
\end{align}
\hfill  
\end{lemma}
\begin{proof}
Following the exact same steps as in the proof of \eqref{int_bound}, we have 
\begin{align}
\bigg| \int_{\hat{z}(x)}^{z(x)}\bar{\alpha}(s)ds\bigg| \leq L_{\alpha}(|e(x)|) |e(x)|^2+|\bar{\alpha}(z(x))||e(x)|. \label{int_bound2}
\end{align}
As a result,
\begin{align}
\bigg| \int_{0}^{1}&w(x)l_y(x,x)\left(\int_{\hat{z}(x)}^{z(x)}\bar{\alpha}(s)ds\right)dx\bigg| \nonumber\\
\leq&~ \left(L_{\alpha}(|e|_{\infty}) |e|_{\infty}+|\bar{\alpha}(z)|_{\infty}\right) \nonumber \\
&~\times \int_{0}^{1}|l_y(x,x)||w(x)||e(x)|dx \nonumber \\
\leq&~ \left(L_{\alpha}(|e|_{\infty}) |e|_{\infty}+|\bar{\alpha}(z)|_{\infty}\right)\sup_{x\in (0,1)}\{|l_y(x,x)|\}  \nonumber \\
&~\times \int_{0}^{1}|w(x)e(x)|dx.\nonumber \\
\leq&~ \left(L_{\alpha}(|e|_{\infty}) |e|_{\infty}+|\bar{\alpha}(z)|_{\infty}\right)\sup_{x\in (0,1)}\{|l_y(x,x)|\} \nonumber \\
&~\times \left( V+\frac{1}{2}|e|_{L^2}^2\right). \nonumber
\end{align}
Using \eqref{v_L2}, we obtain \eqref{F_B_7}. 
\end{proof}

\begin{lemma}
It holds that 
\begin{align}
&\bigg|\bar{\alpha}(\hat{z}(1))p_{10}e(1)\int_{0}^{1}w(x)l(x,1)dx\bigg|\nonumber\\
&\leq \frac{3\sigma}{2a} \left(\bar{\delta}+L_{\alpha}(|w(1)|)|w(1)|\right)|l(\cdot,1)|_{L^2}V \nonumber \\
&+\frac{\sigma}{2a} \left(\bar{\delta}+L_{\alpha}(|w(1)|)|w(1)|\right)|l(\cdot,1)|_{L^2}|w_x|_{L^2}^2.\label{F_B_8} 
\end{align}
\hfill  
\end{lemma}
\begin{proof}
Using the fact that $p_{10}=\sigma/(2a)$, $e(1)=w(1)$, and inequality \eqref{vV}, we have 
\begin{align}
\bigg|\bar{\alpha}(\hat{z}(1))&p_{10}e(1)\int_{0}^{1}w(x)l(x,1)dx\bigg| \nonumber\\
\leq&~ \frac{\sigma}{2a} |\bar{\alpha}(z(1)-e(1))||e(1)||l(\cdot,1)|_{L^2}\sqrt{2V} \nonumber\\
\leq&~ \frac{3\sigma}{2a} |\bar{\alpha}(z(1)-w(1))||l(\cdot,1)|_{L^2}V \nonumber \\
&~+\frac{\sigma}{2a} |\bar{\alpha}(z(1)-w(1))||l(\cdot,1)|_{L^2}|w_x|_{L^2}^2.\nonumber 
\end{align}
Furthermore, note that we have 
\begin{align*}
|\bar{\alpha}(z(1)-w(1))-\bar{\alpha}(z(1))|\leq L_{\alpha}(|w(1)|)|w(1)|. 
\end{align*}
Hence, 
\begin{align*}
|\bar{\alpha}(z(1)-w(1))|\leq |\bar{\alpha}(z(1))|+L_{\alpha}(|w(1)|)|w(1)|.
\end{align*}
As a result, we obtain \eqref{F_B_8}. 
\end{proof}

\begin{lemma}
It holds that 
\begin{align}
&\bigg|\int_{0}^{1}w(x)\int_{x}^{1}l_{yy}(x,y)\int_{\hat{z}(y)}^{z(y)}\bar{\alpha}(s)dsdydx\bigg| \nonumber \\
&\leq 2\sqrt{2}|l_{yy}|_{L^2}\left(\sqrt{1+|p|_{L^2}^2}\right) \left(L_{\alpha}(|e|_{\infty}) |e|_{\infty}+\bar{\delta} \right) V. \label{F_B_9}
\end{align}
\hfill  
\end{lemma}
\begin{proof}
Using \eqref{int_bound2}, note that we have 
\begin{align*}
\bigg|\int_{0}^{1}&w(x)\int_{x}^{1}l_{yy}(x,y)\int_{\hat{z}(y)}^{z(y)}\bar{\alpha}(s)dsdydx\bigg| \\
\leq&~ \int_{0}^{1}|w(x)|\int_{x}^{1}|l_{yy}(x,y)| \left|\int_{\hat{z}(y)}^{z(y)}\bar{\alpha}(s)ds\right|dydx\\
\leq&~ \int_{0}^{1}|w(x)|\int_{x}^{1}|l_{yy}(x,y)| \bigg(L_{\alpha}(|e(y)|) |e(y)|^2 \\
&~+|\bar{\alpha}(z(y))||e(y)|\bigg)dydx \\
\leq&~ L_{\alpha}(|e|_{\infty})|e|_{\infty}\int_{0}^{1}|w(x)|\int_{x}^{1}|l_{yy}(x,y)| |e(y)|dydx \\
&~+ |\bar{\alpha}(z)|_{\infty}\int_{0}^{1}|w(x)|\int_{x}^{1}|l_{yy}(x,y)| |e(y)|dydx \\
\leq&~ L_{\alpha}(|e|_{\infty}) |e|_{\infty}  |e|_{L^2}\int_{0}^{1}|w(x)|\left(\sqrt{\int_{x}^{1}l_{yy}(x,y)^2dy}\right)dx
\end{align*}
\begin{align*}
&~+|\bar{\alpha}(z)|_{\infty} |e|_{L^2}\int_{0}^{1}|w(x)| \left(\sqrt{\int_{x}^{1}l_{yy}(x,y)^2dy}\right)dx \\
\leq&~ \left(L_{\alpha}(|e|_{\infty}) |e|_{\infty}+|\bar{\alpha}(z)|_{\infty} \right)|l_{yy}|_{L^2}|e|_{L^2}\sqrt{2V}.
\end{align*}
Consequently, using \eqref{v_L2}, we obtain \eqref{F_B_9}. 
\end{proof}

We can now derive the following key inequality on $\dot{V}$. 
\begin{proposition}
It holds that 
\begin{align}
&\dot{V} \leq \bigg\{ -2\sigma + \gamma_1 L_{\alpha}\bigg((1+|p|_{\infty})|w|_{\infty}\bigg)(1+|p|_{\infty})|w|_{\infty} \nonumber \\
&+ \gamma_2L_{\alpha}\bigg((1+|p|_{\infty})|w|_{\infty}\bigg) + \gamma_3L_{\alpha}(|w(1)|)|w(1)| \nonumber \\
&+ (\gamma_1+\gamma_3)\bar{\delta}\bigg\}V +\bigg\{-a + \gamma_4L_{\alpha}\bigg((1+|p|_{\infty})|w|_{\infty}\bigg) \nonumber \\
&\times (1+|p|_{\infty})|w|_{\infty} + \gamma_5L_{\alpha}\bigg((1+|p|_{\infty})|w|_{\infty}\bigg)  \nonumber \\
&+ \gamma_6L_{\alpha}(|w(1)|)|w(1)| + (\gamma_4+\gamma_6)\bar{\delta}\bigg\}|w_x|_{L^2}^2.   \label{To_Prove_1}
\end{align}
\hfill  
\end{proposition}
\begin{proof}
Differentiating $V$ with respect to $t$, we obtain 
\begin{align}
\dot{V} =&~ \int_0^1 w\left[aw_{xx}-\sigma w\right]dx +\int_{0}^{1}wfdx. 
\end{align}
Using integration by parts and the boundary conditions $w_x(0)=w_x(1)=0$, note that we have 
\begin{align}
\int_{0}^{1}w(x)w_{xx}(x)dx =&~ \left[w(x)w_x(x)\right]_{x=0}^{x=1} -|w_x|_{L^2}^2 \nonumber \\
=&~-|w_x|_{L^2}^2.
\end{align}
As a result, 
\begin{align}
\dot{V} = - 2 \sigma V - a |w_x|_{L^2}^2+ \int_{0}^{1}w(x)f(x)dx. \label{diff_g}
\end{align}
We will derive now an upperbound on $\int_{0}^{1}w(x)f(x)dx$. To do so, we start by applying integration by parts to the integral terms in the expression of $f$, to obtain 
\begin{align*}
&\int_{x}^{1}l(x,y)\big(\bar{\alpha}(z(y))z_y(y)\big)_ydy = \left[l(x,y)\bar{\alpha}(z(y))z_y(y)\right]_{y=x}^{y=1} \nonumber \\
&~\qquad \qquad \qquad-\int_{x}^{1}l_y(x,y)\bar{\alpha}(z(y))z_y(y)dy \nonumber \\
&= - l(x,x)\bar{\alpha}(z(x))z_x(x) -\left[l_y(x,y)\int_{0}^{z(y)}\bar{\alpha}(s)ds\right]_{y=x}^{y=1} \nonumber \\
&+\int_{x}^{1}l_{yy}(x,y)\int_{0}^{z(y)}\bar{\alpha}(s)ds dy \\
&= - l(x,x)\bar{\alpha}(z(x))z_x(x) - l_y(x,1)\int_{0}^{z(1)}\bar{\alpha}(s)ds \nonumber \\
&+l_y(x,x)\int_{0}^{z(x)}\bar{\alpha}(s)ds+\int_{x}^{1}l_{yy}(x,y)\int_{0}^{z(y)}\bar{\alpha}(s)ds dy.
\end{align*}
Similarly, 
\begin{align*}
&\int_{x}^{1}l(x,y)l(x,y)\big(\bar{\alpha}(\hat{z}(y))\hat{z}_y(y)\big)_ydy \\
&= \left[l(x,y)\bar{\alpha}(z(y))\hat{z}_y(y)\right]_{y=x}^{y=1} -\int_{x}^{1}l_y(x,y)\bar{\alpha}(\hat{z}(y))\hat{z}_y(y)dy \nonumber \\
&= l(x,1)\bar{\alpha}(\hat{z}(1))p_{10}e(1)- l(x,x)\bar{\alpha}(\hat{z}(x))\hat{z}_x(x) \nonumber 
\end{align*}
\begin{align*}
&-\left[l_y(x,y)\int_{0}^{\hat{z}(y)}\bar{\alpha}(s)ds\right]_{y=x}^{y=1} \\
&+\int_{x}^{1}l_{yy}(x,y)\int_{0}^{\hat{z}(y)}\bar{\alpha}(s)ds dy \\
&= l(x,1)\bar{\alpha}(\hat{z}(1))p_{10}e(1) - l(x,x)\bar{\alpha}(\hat{z}(x))\hat{z}_x(x) \nonumber \\
&- l_y(x,1)\int_{0}^{\hat{z}(1)}\bar{\alpha}(s)ds+l_y(x,x)\int_{0}^{\hat{z}(x)}\bar{\alpha}(s)ds \\
&+\int_{x}^{1}l_{yy}(x,y)\int_{0}^{\hat{z}(y)}\bar{\alpha}(s)ds dy.
\end{align*}
As a result, we can rewrite $f$ as 
\begin{align*}
&f(x) = (\bar{\alpha}(z(x))z_x(x))_x-(\bar{\alpha}(\hat{z}(x))\hat{z}_x(x))_x \\
& - l(x,x)\bar{\alpha}(z(x))z_x(x) - l_y(x,1)\int_{0}^{z(1)}\bar{\alpha}(s)ds \nonumber \\
&+l_y(x,x)\int_{0}^{z(x)}\bar{\alpha}(s)ds+\int_{x}^{1}l_{yy}(x,y)\int_{0}^{z(y)}\bar{\alpha}(s)ds dy \\
& -l(x,1)\bar{\alpha}(\hat{z}(1))p_{10}e(1) + l(x,x)\bar{\alpha}(\hat{z}(x))\hat{z}_x(x)\\
&+ l_y(x,1)\int_{0}^{\hat{z}(1)}\bar{\alpha}(s)ds-l_y(x,x)\int_{0}^{\hat{z}(x)}\bar{\alpha}(s)ds \\
&-\int_{x}^{1}l_{yy}(x,y)\int_{0}^{\hat{z}(y)}\bar{\alpha}(s)ds dy. 
\end{align*}
Grouping the terms together we obtain 
\begin{align*}
&f(x) = (\bar{\alpha}(z(x))z_x(x))_x-(\bar{\alpha}(\hat{z}(x))\hat{z}_x(x))_x \\
&~-l(x,x)\bigg(\bar{\alpha}(z(x))z_x(x)-\bar{\alpha}(\hat{z}(x))\hat{z}_x(x)\bigg)\\
&~-l_y(x,1)\int_{\hat{z}(1)}^{z(1)}\bar{\alpha}(s)ds +l_y(x,x)\int_{\hat{z}(x)}^{z(x)}\bar{\alpha}(s)ds \\
&~-l(x,1)\bar{\alpha}(\hat{z}(1))p_{10}e(1) +\int_{x}^{1}l_{yy}(x,y)\int_{\hat{z}(y)}^{z(y)}\bar{\alpha}(s)dsdy. 
\end{align*}
Next, note that 
\begin{align*}
\bar{\alpha}(z)z_x-\bar{\alpha}(\hat{z})\hat{z}_x &= \bar{\alpha}(z)e_x +\left(\bar{\alpha}(z)-\bar{\alpha}(\hat{z})\right)\hat{z}_x \\
&= \bar{\alpha}(z)e_x +\left(\bar{\alpha}(z)-\bar{\alpha}(\hat{z})\right)(z_x-e_x).
\end{align*}
Hence, 
\begin{align*}
&f(x) = \big(\left(\bar{\alpha}(z(x))-\bar{\alpha}(\hat{z}(x))\right)(z_x(x)-e_x(x))\big)_x \\
&~+ \left(\bar{\alpha}(z(x))e_x(x)\right)_x-l(x,x)\bar{\alpha}(z(x))e_x(x)\\
&~-l(x,x)\left(\bar{\alpha}(z(x))-\bar{\alpha}(\hat{z}(x))\right)(z_x(x)-e_x(x))
\end{align*}
\begin{align*}
&~-l_y(x,1)\int_{\hat{z}(1)}^{z(1)}\bar{\alpha}(s)ds +l_y(x,x)\int_{\hat{z}(x)}^{z(x)}\bar{\alpha}(s)ds \\
&~-l(x,1)\bar{\alpha}(\hat{z}(1))p_{10}e(1) +\int_{x}^{1}l_{yy}(x,y)\int_{\hat{z}(y)}^{z(y)}\bar{\alpha}(s)dsdy. 
\end{align*}
We can now derive an upperbound on $\int_{0}^{1}w(x)f(x)dx$. Indeed, combining \eqref{diff_g}, \eqref{F_B_1}, \eqref{F_B_2}, \eqref{F_B_3}, \eqref{F_B_4}, \eqref{F_B_5}, \eqref{F_B_6}, \eqref{F_B_7}, \eqref{F_B_8}, and \eqref{F_B_9}, we conclude on \eqref{To_Prove_1} by noting that 
\begin{align}
|e|_{\infty} \leq (1+|p|_{\infty})|w|_{\infty}. \label{v_max_bound}
\end{align}
\end{proof}

According to \eqref{To_Prove_1}, we need to analyze $|w|_{\infty}$ to study $V$. To this end, we augment $V$ with $H$, since Agmon's inequality provides an upperbound for the max norm of $w$ in terms of its $H^1$ norm. 

In the following lemmas, we derive auxiliary upperbounds that will be used to estimate $\dot{H}$.

\begin{lemma}
It holds that 
\begin{align}
|\big(\bar{\alpha}(z)&-\bar{\alpha}(\hat{z})\big)\big(z_{xx}-e_{xx}\big)|_{L^2}^2 \nonumber \\
\leq&~ \bigg\{ 8(M'')^2\left(1+|p|_{L^2}^2\right)L_{\alpha}(|e|_{\infty})^2 \nonumber \\
&~+8\left[|p_{xx}|_{L^2}^2+\sup_{x\in (0,1)}\left\{\left( p_x(x,x)-\frac{\sigma}{2a}\right)^2\right\} \right] \nonumber \\
&~\times L_{\alpha}(|e|_{\infty})^2|e|_{\infty}^2\bigg\}V+\left(\frac{\sigma}{a}L_{\alpha}(|e|_{\infty})|e|_{\infty}\right)^2|w_x|_{L^2}^2 \nonumber \\
&~+\left(2L_{\alpha}(|e|_{\infty})|e|_{\infty}\right)^2|w_{xx}|_{L^2}^2. \label{g_bound_1}
\end{align}
\hfill  
\end{lemma}
\begin{proof}
Note that we have 
\begin{align*}
|\left(\bar{\alpha}(z)-\bar{\alpha}(\hat{z})\right)\left(z_{xx}-e_{xx}\right)|_{L^2}^2& \leq 2|\left(\bar{\alpha}(z)-\bar{\alpha}(\hat{z})\right)z_{xx}|_{L^2}^2 \\
&+2|\left(\bar{\alpha}(z)-\bar{\alpha}(\hat{z})\right)e_{xx}|_{L^2}^2. 
\end{align*}
As a result, 
\begin{align*}
|\left(\bar{\alpha}(z)-\bar{\alpha}(\hat{z})\right)\left(z_{xx}-e_{xx}\right)|_{L^2}^2& \leq 2|z_{xx}|_{\infty}^2L_{\alpha}(|e|_{\infty})^2|e|_{L^2}^2 \\
&+2L_{\alpha}(|e|_{\infty})^2|e|_{\infty}^2|e_{xx}|_{L^2}^2. 
\end{align*}
Using \eqref{f_back}, we find
\begin{align}
|e_{xx}|_{L^2}^2 \leq&~ 4 \left[|p_{xx}|_{L^2}^2+\sup_{x\in (0,1)}\left\{\left( p_x(x,x)-\frac{\sigma}{2a}\right)^2\right\} \right]V \nonumber \\
&~+\frac{\sigma^2}{2a^2}|w_x|_{L^2}^2+2|w_{xx}|_{L^2}^2. \label{z_xx_b}
\end{align}
As a result, by \eqref{v_L2} and \eqref{z_xx_b}, we obtain \eqref{g_bound_1}. 
\end{proof}

\begin{lemma}
It holds that 
\begin{align}
|\bar{\alpha}(z)&e_{xx}|_{L^2}^2\nonumber \\
\leq&~4 \left[|p_{xx}|_{L^2}^2+\sup_{x\in (0,1)}\left\{\left( p_x(x,x)-\frac{\sigma}{2a}\right)^2\right\} \right] \bar{\delta}^2V \nonumber \\
&~+\frac{\sigma^2}{2a^2}\bar{\delta}^2|w_x|_{L^2}^2+2\bar{\delta}^2|w_{xx}|_{L^2}^2.\label{g_bound_2}
\end{align}
\hfill  
\end{lemma}
\begin{proof}
We have 
\begin{align*}
|\bar{\alpha}(z)e_{xx}|_{L^2}^2\leq |\bar{\alpha}(z)|_{\infty}^2|e_{xx}|_{L^2}^2.
\end{align*}
As a result, using \eqref{z_xx_b}, we obtain \eqref{g_bound_2}.
\end{proof}

\begin{lemma}
It holds that
\begin{align}
|\big( \bar{\alpha}'(z)&-\bar{\alpha}'(\hat{z})\big)\big(z_x-e_x\big)^2|_{L^2}^2 \nonumber \\
\leq&~\left(1+|p|_{L^2}^2\right)\left(4(M')^2L_{\alpha'}(|e|_{\infty})\right)^2V \nonumber \\
&~+\left\{\frac{1}{4}\left(\frac{\sigma}{a}\right)^4+4|p_x|_{L^4}^4\right\}\left(2L_{\alpha'}(|e|_{\infty})|e|_{\infty}\right)^2|w|_{L^4}^4 \nonumber \\
&~+\left(4L_{\alpha'}(|e|_{\infty})|e|_{\infty}\right)^2|w_x|_{L^4}^4. \label{g_bound_3}
\end{align}
\hfill  
\end{lemma}
\begin{proof}
Note that we have 
\begin{align*}
|\big( \bar{\alpha}'(z)&-\bar{\alpha}'(\hat{z})\big)\big(z_x-e_x\big)^2|_{L^2}^2 \\
\leq&~4|\big( \bar{\alpha}'(z)-\bar{\alpha}'(\hat{z})\big)z_x^2|_{L^2}^2+4|\big( \bar{\alpha}'(z)-\bar{\alpha}'(\hat{z})\big)e_x^2|_{L^2}^2 \\
\leq&~\left(2|z_x|_{\infty}^2L_{\alpha'}(|e|_{\infty})\right)^2|e|_{L^2}^2 +\left(2L_{\alpha'}(|e|_{\infty})|e|_{\infty}\right)^2|e_x|_{L^4}^4. 
\end{align*}
Furthermore, note that 
\begin{align}
|e_x|_{L^4}^4 \leq&~ 4|w_x|_{L^4}^4+\frac{1}{4}\left(\frac{\sigma}{a}\right)^4|w|_{L^4}^4  \nonumber \\
&~+ 4\int_{0}^{1}\left(\int_{x}^{1}p_x(x,y)w(y)dy\right)^4dx \nonumber \\
\leq&~ 4|w_x|_{L^4}^4+\frac{1}{4}\left(\frac{\sigma}{a}\right)^4|w|_{L^4}^4  \nonumber \\
&~+ 4\int_{0}^{1}\left(\int_{x}^{1}p_x(x,y)^2dy\right)^2\left(\int_{x}^{1}w(y)^2dy\right)^2dx \nonumber \\
\leq&~4|w_x|_{L^4}^4+\left\{\frac{1}{4}\left(\frac{\sigma}{a}\right)^4+4|p_x|_{L^4}^4\right\}|w|_{L^4}^4.\label{z_x_4_b}
\end{align}
Consequently, using \eqref{v_L2} and \eqref{z_x_4_b}, we obtain \eqref{g_bound_3}. 
\end{proof}

\begin{lemma}
It holds that
\begin{align}
|2z_x\bar{\alpha}(z)e_x|_{L^2}^2 \leq&~ 12(M')^2\left(\frac{\sigma^2}{2a^2}+2|p_x|_{L^2}^2\right)\bar{\delta}^2V \nonumber \\
&~+12(M')^2\bar{\delta}^2|w_x|_{L^2}^2,\label{g_bound_4}\\
|\bar{\alpha}(z)e_x^2|_{L^2}^2 \leq&~ \left\{\frac{1}{4}\left(\frac{\sigma}{a}\right)^4+4|p_x|_{L^4}^4\right\}\bar{\delta}^2|w|_{L^4}^4 \nonumber \\
&~+4\bar{\delta}^2|w_x|_{L^4}^4. \label{g_bound_5}
\end{align}
\hfill  
\end{lemma}
\begin{proof}
We have 
\begin{align*}
|2z_x\bar{\alpha}(z)e_x|_{L^2}^2 \leq \left(2|z_x|_{\infty} |\bar{\alpha}(z)|_{\infty}\right)^2|e_x|_{L^2}^2. 
\end{align*}
Using \eqref{vx_L2_B}, we obtain \eqref{g_bound_4}. Similarly, we have 
\begin{align*}
|\bar{\alpha}(z)e_x^2|_{L^2}^2 \leq |\bar{\alpha}(z)|_{\infty}^2|e_x|_{L^4}^4.
\end{align*}
As a result, using \eqref{z_x_4_b}, we obtain \eqref{g_bound_5}. 
\end{proof}

We can now derive the following inequality on $\dot{H}$. 

\begin{proposition}
It holds that 
\begin{align}
&\dot{H} \leq \bigg\{-\frac{a}{2}+\gamma_7L_{\alpha}\bigg((1+|p|_{\infty})|w|_{\infty}\bigg)^2(1+|p|_{\infty})^2|w|_{\infty}^2 \nonumber\\
&+\frac{\gamma_7}{2}\bar{\delta}^2\bigg\}|w_{xx}|_{L^2}^2 + \bigg\{ -\sigma +\gamma_8L_{\alpha}\bigg((1+|p|_{\infty})|w|_{\infty}\bigg)^2\nonumber \\
&\times (1+|p|_{\infty})^2|w|_{\infty}^2+\gamma_9\bar{\delta}^2\bigg\}|w_x|_{L^2}^2 + \bigg\{\gamma_{10}\bar{\delta}^2 \nonumber \end{align}
\begin{align}
&+4\gamma_{10}L_{\alpha'}\bigg((1+|p|_{\infty})|w|_{\infty}\bigg)^2(1+|p|_{\infty})^2|w|_{\infty}^2\bigg\}|w|_{L^4}^4 \nonumber \\
&+\bigg\{ \gamma_{11}L_{\alpha'}\bigg((1+|p|_{\infty})|w|_{\infty}\bigg)^2(1+|p|_{\infty})^2|w|_{\infty}^2 \nonumber \\
&+\frac{\gamma_{11}}{4}\bar{\delta}^2\bigg\}|w_x|_{L^4}^4+\bigg\{\gamma_{12}L_{\alpha}\bigg((1+|p|_{\infty})|w|_{\infty}\bigg)^2 \nonumber \\
&+\gamma_{13}L_{\alpha'}\bigg((1+|p|_{\infty})|w|_{\infty}\bigg)^2+\gamma_{14}\bar{\delta}^2 \nonumber \\
&+\gamma_{15}L_{\alpha}\bigg((1+|p|_{\infty})|w|_{\infty}\bigg)^2(1+|p|_{\infty})^2|w|_{\infty}^2  \bigg\}V.\label{To_Prove_2}
\end{align}
\hfill  
\end{proposition}
\begin{proof}
Suppose that $w_{xt}$ exists and is continuous. Then, one can differentiate $H$ with respect to $t$ to obtain 
\begin{align*}
\dot{H} = \int_{0}^{1}w_x(x)w_{xt}(x)dx. 
\end{align*}
Applying integration by parts and the boundary conditions $w_x(0)=w_x(1)=0$, we obtain 
\begin{align}
\dot{H} &= \left[w_x(x)w_t(x)\right]_{x=0}^{x=1}-\int_{0}^{1}w_{xx}(x)w_t(x)dx \nonumber \\
&= -\int_{0}^{1}w_{xx}(x)w_t(x)dx. \label{H_dot}
\end{align}
Now, $w$ may not be regular enough for $w_{xt}$ to be continuous. Equation \eqref{H_dot} still holds in this case, and it can be established using the classical techniques in, e.g., \cite{ladyzen}. Specifically, we consider the map $(t,h)\mapsto H_h(t)$ defined for $(t,h)\in [0,+\infty)\times (0,+\infty)$ by
$$ H_h(t) := \frac{1}{2}\left|\frac{1}{h}\int_{t}^{t+h}w_x(\cdot,\tau)d\tau \right|_{L^2}^2. $$ 
Using l'Hôpital's rule and the Lebesgue dominated convergence theorem, we first show that 
\begin{align}
\lim_{h\to 0^-} H_h(t) = \frac{1}{2}|w_x(\cdot,t)|_{L^2}^2 \quad \forall t\geq 0. \label{37}
\end{align}
Indeed, both functions $h\mapsto h$ and $h\mapsto \int_t^{t+h}w_x(\cdot,\tau)d\tau$ vanish at $h=0$, hence 
\begin{align*}
\lim_{h\to 0^-}\frac{1}{h}\int_t^{t+h}w_x(\cdot,\tau)d\tau = w_x(\cdot,t).
\end{align*}
The Lebesgue dominated convergence theorem allows us then to take the limit out of the $L^2$ norm, to obtain \eqref{37}. Next, we evaluate $\dot{H}_h(t) := \frac{\partial }{\partial t}H_h(t)$ for all $t\geq 0$, and obtain
\begin{align}
&\dot{H}_h(t) = \bigg[\bigg(\frac{\int_{t}^{t+h}w_x(x,\tau)d\tau}{h}\bigg) \bigg(\frac{w(x,t+h)-w(x,t)}{h}\bigg)\bigg]_{0}^{1} \nonumber \\
&~- \int_{0}^{1}\bigg(\frac{\int_{t}^{t+h}w_{xx}(x,\tau)d\tau}{h}\bigg) \bigg(\frac{w(x,t+h)-w(x,t)}{h}\bigg)dx. \nonumber
\end{align}
Integrating the latter identity with respect to $t$, from $0$ to any $t\geq 0$, and taking the limit as $h\to 0^-$ while using \eqref{37} and the Lebesgue dominated convergence theorem, we can show that
\begin{align}
\frac{1}{2}|w_x(\cdot,t)|_{L^2}^2 & = \int_{0}^{t}\big[w_x(x,s) w_t(x,s)\big]_{0}^{1}ds +\frac{1}{2}|w_o'|_{L^2}^2 \nonumber \\
&~- \int_{0}^{t}\int_{0}^{1}w_{xx}(x,s)w_t(x,s)dxds \quad \forall t\geq 0. \nonumber 
\end{align}
As a result, the map 
$t\mapsto |w_x(\cdot,t)|_{L^2}^2$ is differentiable on $[0,+\infty)$ and \eqref{H_dot} holds.

Next, we use \eqref{H_dot} to derive an upperbound on $\dot{H}$. To do so, first note that we have 
\begin{align*}
\dot{H} =&~ - \int_{0}^{1}w_{xx}(x)\left(aw_{xx}(x)-\sigma w(x)+f(x)\right)dx \\
=&~ - a|w_{xx}|_{L^2}^2 + \sigma \left[w_x(x)w(x) \right]_{x=0}^{x=1}-\sigma |w_x|_{L^2}^2 \nonumber \\
&~-\int_{0}^{1}w_{xx}(x)f(x)dx \\
=&~ - a |w_{xx}|_{L^2}^2-\sigma |w_x|_{L^2}^2-\int_{0}^{1}w_{xx}(x)f(x)dx. 
\end{align*}
Applying Young's inequality, we obtain 
\begin{align}
\dot{H} \leq - \frac{a}{2}|w_{xx}|_{L^2}^2-\sigma |w_x|_{L^2}^2 + \frac{1}{2a}|f|_{L^2}^2. \label{g_dot}
\end{align}
Hence, to upperbound $\dot{H}$, it remains to upperbound $|f|_{L^2}^2$. To this end, note that, according to \eqref{f_def}, we have 
\begin{align*}
f(x) = g(x) + \int_{x}^{1}l(x,y)g(y)dy, 
\end{align*}
where 
\begin{align*}
g(x) := (\bar{\alpha}(z(x))z_x(x))_x-(\bar{\alpha}(\hat{z}(x))\hat{z}_x(x))_x. 
\end{align*}
As a result, applying the Cauchy-Schwarz inequality, we obtain  
\begin{align}
|f|_{L^2}^2 \leq 2\left(1+|l|_{L^2}^2\right)|g|_{L^2}^2. \label{g_bound_0}
\end{align}
Therefore, to upperbound $|f|_{L^2}^2$, it is sufficient to upperbound $|g|_{L^2}^2$. We start by rewriting $g$ as 
\begin{align*}
g =&~ \bar{\alpha}'(v)z_x^2+\bar{\alpha}(z)z_{xx}-\bar{\alpha}'(\hat{z})\hat{z}_x^2-\bar{\alpha}(\hat{z})\hat{z}_{xx}\\
=&~ \left(\bar{\alpha}(z)-\bar{\alpha}(\hat{z})\right)\hat{z}_{xx}+\bar{\alpha}(z)e_{xx} + \bar{\alpha}'(v)v_{x}^2-\bar{\alpha}'(\hat{z})\hat{z}_x^2 \\
=&~ \left(\bar{\alpha}(z)-\bar{\alpha}(\hat{z})\right)\hat{z}_{xx}+\bar{\alpha}(z)e_{xx}+\left( \bar{\alpha}'(z)-\bar{\alpha}'(\hat{z})\right)\hat{z}_x^2 \\
&~+\bar{\alpha}(z)(z_x^2-\hat{z}_x^2) \\
&=\left(\bar{\alpha}(z)-\bar{\alpha}(\hat{z})\right)\left(z_{xx}-e_{xx}\right)+\bar{\alpha}(z)e_{xx}\\
&+\left( \bar{\alpha}'(z)-\bar{\alpha}'(\hat{z})\right)\left(z_x-e_x\right)^2 +2z_x\bar{\alpha}(z)e_x-\bar{\alpha}(z)e_x^2.
\end{align*}
Combining \eqref{g_dot}, \eqref{g_bound_0}, \eqref{g_bound_1}, \eqref{g_bound_2}, \eqref{g_bound_3}, \eqref{g_bound_4}, \eqref{g_bound_5}, and \eqref{v_max_bound}, we conclude on \eqref{To_Prove_2}.  
\end{proof}

Now, we analyze 
\begin{align*}
E(w(\cdot,t)) = V(w(\cdot,t))+H(w(\cdot,t)). 
\end{align*}
Using \eqref{To_Prove_1} and \eqref{To_Prove_2}, we derive the following inequality on $\dot{E}$. 
\begin{lemma}
It holds that 
\begin{align}
\dot{E} \leq&~ \left(-2\sigma + \varepsilon_6\big(E,\bar{\delta}\big) \right)V +\left(-a-\sigma +\varepsilon_7\big(E,\bar{\delta}\big)  \right)|w_x|_{L^2}^2 \nonumber \\
&~+ \left( -\frac{a}{2}+\varepsilon_8\big(E,\bar{\delta}\big)\right)|w_{xx}|_{L^2}^2. \label{E_dot_2}
\end{align}
\hfill  
\end{lemma}
\begin{proof}
First note that, by Agmon's inequality,
\begin{align}
|w|_{\infty} \leq 2\sqrt{E}. \label{agmon_max}
\end{align}
Consequently, since the maps $s\mapsto L_{\alpha}(s)$ and $s\mapsto L_{\alpha'}(s)$ are nondecreasing, then 
\begin{align*}
&L_{\alpha}\left(\left(1+|p|_{\infty}\right)|w|_{\infty}\right) \leq L_{\alpha}\left(2\left(1+|p|_{\infty}\right)\sqrt{E}\right), \\
&L_{\alpha'}\left(\left(1+|p|_{\infty}\right)|w|_{\infty}\right) \leq L_{\alpha'}\left(2\left(1+|p|_{\infty}\right)\sqrt{E}\right).
\end{align*}
As a result,
\begin{align}
\dot{E} \leq&~ \bigg(-2\sigma + \varepsilon_1 + \gamma_3|\bar{\alpha}(z(1))|+\gamma_1|\bar{\alpha}(z)|_{\infty} \nonumber \\
&~+\gamma_{14}|\bar{\alpha}(z)|_{\infty}^2\bigg)V +\left(\varepsilon_2+\gamma_{10}|\bar{\alpha}(z)|_{\infty}^2\right)|w|_{L^4}^4\nonumber \\
&~+\big(-a-\sigma + \varepsilon_3 +\gamma_6 |\bar{\alpha}(z(1))|+\gamma_4|\bar{\alpha}(z)|_{\infty} \nonumber \\
&~+\gamma_9|\bar{\alpha}(z)|_{\infty}^2\big)|w_x|_{L^2}^2+\left(\varepsilon_4+\frac{\gamma_{11}}{4}|\bar{\alpha}(z)|_{\infty}^2\right)|w_x|_{L^4}^4 \nonumber \\
&~+ \left(-\frac{a}{2}+\varepsilon_5+\frac{\gamma_7}{2}|\bar{\alpha}(z)|_{\infty}^2\right)|w_{xx}|_{L^2}^2.  \label{E_dot_1}
\end{align}
Next, note that
\begin{align}
|w|_{L^4}^4 &\leq 2|w|_{\infty}^2V \leq 8EV. \label{to_use_L40}
\end{align}
Similarly, 
\begin{align*}
|w_{x}|_{L^4}^4 \leq |w_x|_{\infty}^2|w_x|_{L^2}^2. 
\end{align*}
Furthermore, by Agmon's inequality,
\begin{align*}
|w_x|_{\infty}^2 \leq 2|w_x|_{L^2}^2+|w_{xx}|_{L^2}^2\leq 4E + |w_{xx}|_{L^2}^2. 
\end{align*}
Hence, 
\begin{align}
|w_x|_{L^4}^4 &\leq 4E |w_x|_{L^2}^2+|w_x|_{L^2}^2|w_{xx}|_{L^2}^2 \nonumber \\
&\leq 4E |w_x|_{L^2}^2 + 2E|w_{xx}|_{L^2}^2. \label{to_use_L4}
\end{align}
Combining \eqref{E_dot_1}, \eqref{to_use_L40}, and \eqref{to_use_L4}, we finally obtain \eqref{E_dot_2}.
\end{proof}

Now, using \eqref{E_dot_2}, we will first show that, if $E(0)$ and $\bar{\delta}$ are small enough, then
\begin{align}
E(t) \leq E(0) \quad \forall t\geq 0, \label{E_decay} 
\end{align}
where we have used the notation $E(t):=E(w(\cdot,t))$. Specifically, we establish the following result. 
\begin{lemma}
Suppose that
\begin{align}
-2\sigma + \varepsilon_6\big(E(0),\bar{\delta}\big) &< 0, \label{R1} \\
-a-\sigma +\varepsilon_7\big(E(0),\bar{\delta}\big) &<0, \label{R2} \\
-\frac{a}{2}+\varepsilon_8\big(E(0),\bar{\delta}\big) &< 0. \label{R3}
\end{align}
Then, \eqref{E_decay} holds.
\hfill  
\end{lemma}
\begin{proof}
According to \eqref{E_dot_2}, to conclude on \eqref{E_decay}, it is sufficient to show that, for all $t \geq 0$,
\begin{equation} \label{E_decay_proof}
\hspace{-0.3cm}\left\lbrace 
\begin{aligned}
-2\sigma + \varepsilon_6\big(E(t),\bar{\delta}\big) &< 0, \\
-a-\sigma +\varepsilon_7\big(E(t),\bar{\delta}\big) &<0, \\
-\frac{a}{2}+\varepsilon_8\big(E(t),\bar{\delta}\big) &< 0. 
\end{aligned}
\right. 
\end{equation}
We already know that \eqref{E_decay_proof} holds at $t=0$. To show that it also holds for all $t\in (0,T)$, we can proceed by contraposition. Indeed, if $E$ increases on some time interval, then we must have $\dot{E}>0$ on that time interval. However, by \eqref{E_dot_2}, this would mean that one of the inequalities in \eqref{E_decay_proof} does not hold at some time instant.
\end{proof}

We can finally establish the main result. 

\begin{theorem}\label{estimates}
Suppose there exists $\omega^*>0$ satisfying
\begin{align}
\varepsilon_6\left(\frac{M_l}{2}(\omega^*)^2,\bar{\delta}\right) &< 2\sigma, \label{omega_cond_1}\\
\varepsilon_7\left(\frac{M_l}{2}(\omega^*)^2,\bar{\delta}\right) &< a+\sigma, \label{omega_cond_2}\\
\varepsilon_8\left(\frac{M_l}{2}(\omega^*)^2,\bar{\delta}\right) &< \frac{a}{2},\label{omega_cond_3}
\end{align}
where 
\begin{align*}
M_l := \max\bigg\{3,2\left(1+|l|_{L^2}^2\right),3\left(\left(\frac{\sigma}{2a}\right)^2+|l_x|_{L^2}^2\right)\bigg\}.
\end{align*}
If $e_o\in \mathscr{E}_o(\omega^*)$, then, for all $t\geq 0$,
\begin{align}
|e(\cdot,t)|_{H^1} &\leq \sqrt{M_pM_l}|e_o|_{H^1}\exp\left(-\frac{\sigma^*}{2}t\right), \label{est_1} \\
|e(\cdot,t)|_{\infty} &\leq \sqrt{2M_pM_l}|e_o|_{H^1}\exp\left(-\frac{\sigma^*}{2}t\right), \label{est_2}
\end{align}
where
\begin{align}
\sigma^*&:=\min\left\{2\sigma-\varepsilon_6\big(E^*,\bar{\delta}\big),\right. \nonumber \\
&\left.\qquad \qquad 2(a+\sigma)-2\varepsilon_7\big(E^*,\bar{\delta}\big)\right\}>0, \label{conv_rate}\\
E^* &:= \frac{M_l}{2}(\omega^*)^2, \nonumber \\
M_p &:= \max\bigg\{3,2\left(1+|p|_{L^2}^2\right),3\left(\left(\frac{\sigma}{2a}\right)^2+|p_x|_{L^2}^2\right)\bigg\}. \nonumber 
\end{align}
\end{theorem}
\begin{proof}
Note that if $e_o\in \mathscr{E}_o(\omega^*)$, then \eqref{R1}-\eqref{R2}-\eqref{R3} hold. Hence, combining \eqref{E_dot_2} and \eqref{E_decay}, we have 
\begin{align}
\dot{E} \leq - \sigma^* E, 
\end{align}
which implies that 
\begin{align}
|w(\cdot,t)|_{H^1} \leq |w_o|_{H^1}\exp\left(-\frac{\sigma^*}{2}t\right) \quad \forall t\geq 0. \label{H1_w}
\end{align}
Furthermore, using \eqref{v_L2} and \eqref{vx_L2_B}, we have 
\begin{align}
|e|_{H^1} \leq \sqrt{M_p}|w|_{H^1}.
\end{align}
Using \eqref{inv_back}, one can also show that 
\begin{align}
|w|_{H^1} \leq \sqrt{M_l}|e|_{H^1}.
\end{align}
As a result, we conclude from \eqref{H1_w} on \eqref{est_1}. 
Finally, applying Agmon's inequality stating that $|e|_{\infty}\leq \sqrt{2}|e|_{H^1}$, we obtain \eqref{est_2}.
\end{proof}

\begin{remark}
A necessary and sufficient condition for the existence of $\omega^*>0$ such that \eqref{omega_cond_1}-\eqref{omega_cond_2}-\eqref{omega_cond_3} hold is 
\begin{align}
(\gamma_1+\gamma_3)\bar{\delta}+\gamma_{14}\bar{\delta}^2 &< 2\sigma, \label{cond_1} \\
(\gamma_4+\gamma_6)\bar{\delta}+\gamma_9\bar{\delta}^2 &< a+\sigma, \label{cond_2}\\
\frac{\gamma_7}{2}\bar{\delta}^2 &< \frac{a}{2}. \label{cond_3}
\end{align}
This is due to the fact that $\varepsilon_6$, $\varepsilon_7$, and $\varepsilon_8$ are continuous in $E$, and 
\begin{align*}
\varepsilon_6(0,\bar{\delta}) &= (\gamma_1+\gamma_3)\bar{\delta}+\gamma_{14}\bar{\delta}^2, \\
\varepsilon_7(0,\bar{\delta}) &= (\gamma_4+\gamma_6)\bar{\delta}+\gamma_9\bar{\delta}^2, \\
\varepsilon_8(0,\bar{\delta}) &= \frac{\gamma_7}{2}\bar{\delta}^2. 
\end{align*}
It is important to note that conditions \eqref{cond_1}-\eqref{cond_2}-\eqref{cond_3} cannot, in general, be ensured simply by choosing $\sigma$ and $a$ sufficiently large, since the constants $\{\gamma_i\}$ for $i\in \{1,3,4,6,7\}$ themselves depend on $\sigma$ and $a$; see Table \ref{tab1}. Nevertheless, if $\bar{\delta}$ is sufficiently small, then the gains $\sigma$ and $a$ can be selected so that \eqref{cond_1}-\eqref{cond_2}-\eqref{cond_3} hold, because $\{\gamma_i\}$ for $i\in \{1,3,4,6,7\}$ do not depend on $\bar{\delta}$.  

It can be noted that when $\alpha = a$ (the linear case), we have $\bar{\alpha} = 0$ and hence $\bar{\delta}=0$. In this case, conditions \eqref{cond_1}--\eqref{cond_3} are trivially satisfied. Furthermore, we have $\varepsilon_i = 0$ for each $i\in \{6,7,8\}$ and $\sigma^*=\sigma$. Hence, \eqref{omega_cond_1}--\eqref{omega_cond_3} hold for any $\omega^*>0$, and we recover the well-known global exponential stability result of the linear case at rate $\sigma$.
\hfill $\bullet$
\end{remark}

\section{Example}\label{example}

The dependence of the radius of the region of attraction $\omega^*>0$ on the various parameters of the system is not trivial and is described by the nonlinear relations in \eqref{omega_cond_1}--\eqref{omega_cond_3}. Similarly, the relation between the decay rate $\sigma^*$ and $\sigma$ is not trivial and, unlike in the linear case, where $\sigma^*=\sigma$ with $\sigma$ a free design parameter, need not be monotonic in the nonlinear case. Instead, it is described by the nonlinear relations established in this article. The objective of this section is to illustrate these non-monotonic dependencies through an example.

For the simulations, to mimic the fact that $\Sigma$ represents the physical continuous system, whereas $\hat{\Sigma}$ is the system implemented numerically, we use a very fine mesh to discretize $\Sigma$ and a much coarser mesh to discretize $\hat{\Sigma}$. Accordingly, we consider two discretizations of the interval $[0,1]$: a set of points $\{x_i\}_{i=0}^{N}$ with $N\in \mathbb{N}^*$, $x_i=ih$ and $h=1/N$; and another set of points $\{\hat{x}_i\}_{i=0}^{\hat{N}}$ with $\hat{N} \in \mathbb{N}^*$, $\hat{x}_i = i \hat{h}$ and $\hat{h} = 1/\hat{N}$. We also consider two discretizations of the time interval $[0,t_f]$, where $t_f>0$ is the final time of simulation: a set of time instants $\{t_n\}_{n=0}^{N_t}$ with $N_t\in \mathbb{N}^*$, $t_n = n \Delta t$ and $\Delta t = t_f/N_t$; and another set of time instants $\{\hat{t}_k\}_{k=0}^{\hat{N}_t}$ with $\hat{N}_t = \frac{1}{r}N_t$ and $r\in \mathbb{N}^*$, $\hat{t}_k = k\widehat{\Delta t}$ and $\widehat{\Delta t} = t_f/\hat{N}_t = r\Delta t$. We discretize $\Sigma$ on $\{x_i\}_{i=0}^N\times \{t_n\}_{n=0}^{N_t}$, and we discretize $\hat{\Sigma}$ on $\{\hat{x}_i\}_{i=0}^{\hat{N}}\times \{\hat{t}_k\}_{k=0}^{\hat{N}_t}$. 

We start by describing the numerical scheme used to simulate $\Sigma$. We need to approximate 
$$\left(\alpha(z)z_x\right)_x = \alpha(z)_x z_x + \alpha(z)z_{xx}.$$
To this end, we employ a central difference scheme to approximate $\alpha (z(x_i))z_{xx}(x_i)$, i.e., 
\begin{align}
\alpha(z(x_i))z_{xx}(x_i) \approx \alpha(z_i)&\frac{z_{i+1}-2z_i+z_{i-1}}{h^2} \nonumber \\
&~ \forall i\in \{1,2,...,N-1\}, \label{approx_1}
\end{align}
where $z_i \approx z(x_i)$. Furthermore, we employ an Euler backward scheme to approximate $\alpha(z)_xz_x\big|_{x=x_i}$, i.e., 
\begin{align}
\alpha(z)_xz_x\big|_{x=x_i} \approx &\frac{(\alpha(z_i)-\alpha(z_{i-1}))(z_i-z_{i-1})}{h^2} \nonumber \\
&~\quad \forall i\in \{1,2,...,N-1\}.  \label{approx_2}
\end{align}
Combining \eqref{approx_1} and \eqref{approx_2}, we obtain 
\begin{align}
\left(\alpha(z)z_x\right)_x \big|_{x=x_i} \approx \alpha(z_i)&\frac{z_{i+1}-z_i}{h^2}-\alpha(z_{i-1})\frac{z_i-z_{i-1}}{h^2}, \nonumber \\
&\quad \forall i\in \{1,2,...,N-1\}. \nonumber 
\end{align}
To approximate the boundary conditions $z_x(0,t)=z_x(1,t)=0$, we let 
\begin{align*}
z_0 = z_1 \quad \text{and} \quad z_N=z_{N-1}. 
\end{align*}

For time discretization, we use a semi-implicit Euler backward scheme. That is, we let
\begin{align*}
\frac{z_i^{n+1}-z_i^n}{\Delta t} =&~ \alpha(z_i^n)\frac{z^{n+1}_{i+1}-z^{n+1}_i}{h^2}-\alpha(z^n_{i-1})\frac{z^{n+1}_i-z^{n+1}_{i-1}}{h^2}, \\
&~ \quad \forall i\in \{1,2,...,N-1\},
\end{align*}
where $z_i^n \approx z_i(t_n)$.

As a result, $\Sigma$ is approximated by the tridiagonal system
\begin{equation*}
\Sigma_d : \left\lbrace 
\begin{aligned}
&\begin{bmatrix}
A_1(z_i^n) & A_2(z_i^n,z_{i-1}^n) & A_1(z_{i-1}^n)
\end{bmatrix} \begin{bmatrix}
z_{i+1}^{n+1} \\ z_i^{n+1} \\ z_{i-1}^{n+1}
\end{bmatrix} = z_i^n \\
&\qquad \qquad \qquad \qquad \qquad \qquad \forall i\in \{1,2,...,N-1\}, \\
&z_0^{n+1} = z_1^{n+1}, \\
&z_N^{n+1} = z_{N-1}^{n+1},
\end{aligned}
\right.
\end{equation*}
where 
\begin{align*}
A_1(z_i^n) &:= -\frac{\Delta t}{h^2}\alpha (z_i^n),  \\
A_2(z_i^n,z_{i-1}^n) &:= 1 + \frac{\Delta t}{h^2}\left(\alpha (z_i^n) + \alpha (z_{i-1}^n)\right).
\end{align*}
This system is solved using the Thomas algorithm.

We turn now to the approximation of $\hat{\Sigma}$. We need to approximate the in-domain injection term $p_1(x)(z(1,t)-\hat{z}(1,t))$ and the boundary injection term $p_{10}(z(1,t)-\hat{z}(1,t))$. For $p_1(x)$, note that, from its definition in \eqref{p1} and the definition of $p$ in \eqref{kernel_1}, we have 
\begin{align}
p_1(x)
=&~ \left(\frac{\sigma}{\sqrt{\frac{\sigma}{a}\left(1-x^2\right)}}-\frac{\sqrt{a\sigma}}{(1-x^2)^{3/2}}\right) \nonumber \\
&~\times \sum_{m=0}^{+\infty} \frac{1}{2^{1+2m}m!(m+1)!}\left(\sqrt{\frac{\sigma}{a}(1-x^2)}\right)^{1+2m} \nonumber \\
&~+\frac{\sigma}{1-x^2}\sum_{m=0}^{+\infty}\frac{1+2m}{2^{1+2m}m!(m+1)!}\left(\frac{\sigma}{a}(1-x^2)\right)^{m}.\label{p_1_approx}
\end{align}
We approximate $p_1(x_j)\approx \hat{p}_j$, where $\hat{p}_j$ is given by the right-hand side of \eqref{p_1_approx} with $x=x_j$ and the infinite series truncated to a finite number of terms. That is, we replace $\sum_{m=0}^{+\infty}$ by $\sum_{m=0}^{q}$ for some $q\in \mathbb{N}$. 

Since $\widehat{\Delta t}$ is a multiple of $\Delta t$, we can write
\begin{align*}
\frac{\hat{z}_j^{k+1}-z_j^k}{\widehat{\Delta t}} =&~ \alpha(\hat{z}_j^k)\frac{\hat{z}^{k+1}_{j+1}-\hat{z}^{k+1}_j}{\hat{h}^2}-\alpha(\hat{z}^k_{j-1})\frac{\hat{z}^{k+1}_j-\hat{z}^{k+1}_{j-1}}{\hat{h}^2}, \\
&~+\hat{p}_j\left(z_{N}^{kr}-\hat{z}_{\hat{N}}^{k}\right)\quad \forall j\in \{1,2,...,\hat{N}-1\},
\end{align*}
where $\hat{z}_j^k \approx \hat{z}(\hat{x}_j,\hat{t}_k)$.

The boundary condition $\hat{z}_x(0,t)=0$ is approximated by setting 
$$ \hat{z}_0^{k+1} = \hat{z}_1^{k+1}.$$
The boundary condition $\hat{z}_x(1,t) = p_{10}\left(z(1,t)-\hat{z}(1,t)\right)$ is approximated by 
\begin{align*}
\frac{\hat{z}_{\hat{N}}^{k+1}-\hat{z}_{\hat{N}-1}^{k+1}}{\hat{h}} = p_{10}\left(z_{N}^{r(k+1)}-\hat{z}_{\hat{N}}^{k+1}\right). 
\end{align*}
As a result, $\hat{\Sigma}$ is approximated by the tridiagonal system 
\begin{equation*}
\hat{\Sigma}_d : \left\lbrace 
\begin{aligned}
&\begin{bmatrix}
A_1(\hat{z}_j^k) & A_2(\hat{z}_j^k,\hat{z}_{j-1}^k) & A_1(\hat{z}_{j-1}^k)
\end{bmatrix} \begin{bmatrix}
\hat{z}_{j+1}^{k+1} \\ \hat{z}_j^{k+1} \\ \hat{z}_{j-1}^{k+1}
\end{bmatrix} = \hat{z}_j^k \\
&~\qquad \quad +\widehat{\Delta t}\hat{p}_j\left(z_{N}^{kr}-\hat{z}_{\hat{N}}^{k}\right) \quad \forall j\in \{1,2,...,\hat{N}-1\}, \\
&\hat{z}_0^{k+1} = \hat{z}_1^{k+1}, \\
&\hat{z}_{\hat{N}}^{k+1} = \frac{\hat{h}p_{10}}{1+\hat{h}p_{10}}z_N^{r(k+1)}+\frac{1}{1+\hat{h}p_{10}}\hat{z}_{\hat{N}-1}^{k+1},
\end{aligned}
\right.
\end{equation*}
that is solved using the Thomas algorithm. 

We consider the system $\Sigma$ with 
\begin{equation*}
\alpha (z) := \left(z^2+\ell \right)^{5/4} \quad \forall z\in \mathbb{R},
\end{equation*}
for some constant $\ell>0$. This is a very simplified model of ion temperature along a magnetic field line in the scrape-off layer of a tokamak \cite[Pages 652--684]{plasma}. The constant $\ell$, intended to be very small, is introduced solely to regularize $\alpha$ near $z=0$ and ensure that $\alpha \in \mathscr{C}^2(\mathbb{R};\mathbb{R}_{>0})$, as required by Assumption \ref{ass_1}. Accordingly, the Lipschiz constants in \eqref{lipschitz1}-\eqref{lipschitz2} are 
\begin{align*}
L_{\alpha}(s) &= \frac{5}{2}s\left(s^2+\ell\right)^{1/4}, \\
L_{\alpha'}(s) &= \frac{5}{2}\left(s^2+\ell\right)^{1/4} + \frac{5}{4}\frac{s^2}{(s^2+\ell)^{3/4}},
\end{align*}
for all $s\geq 0$. We set $\ell = 10^{-3}$.

Furthermore, we select the initial condition 
$$z_o(x) := 5+0.3 \cos(\pi x) \quad \forall x\in [0,1].$$
Note that $z_o \in \mathscr{H}^{2+\beta}([0,1])$ for any $\beta\in (0,1)$, and $z_o'(0)=z_o'(1)=0$. Hence, Assumption \ref{ass_1} holds and $\Sigma$ admits a unique forward complete classical solution. 

In Figure \ref{fig1}, we simulate $\Sigma_d$ with $h=1/80$, $\Delta t = 10^{-5}$, and $t_f = 5 \times 10^{-3}$. The time interval has been kept relatively short to capture the transient behavior, as the nonlinearity $\alpha$ induces rapid diffusion, causing the state to quickly settle to a constant value equal to the spatial average of the initial condition. 
\begin{figure}
    \centering
    \includegraphics[width=0.7\linewidth]{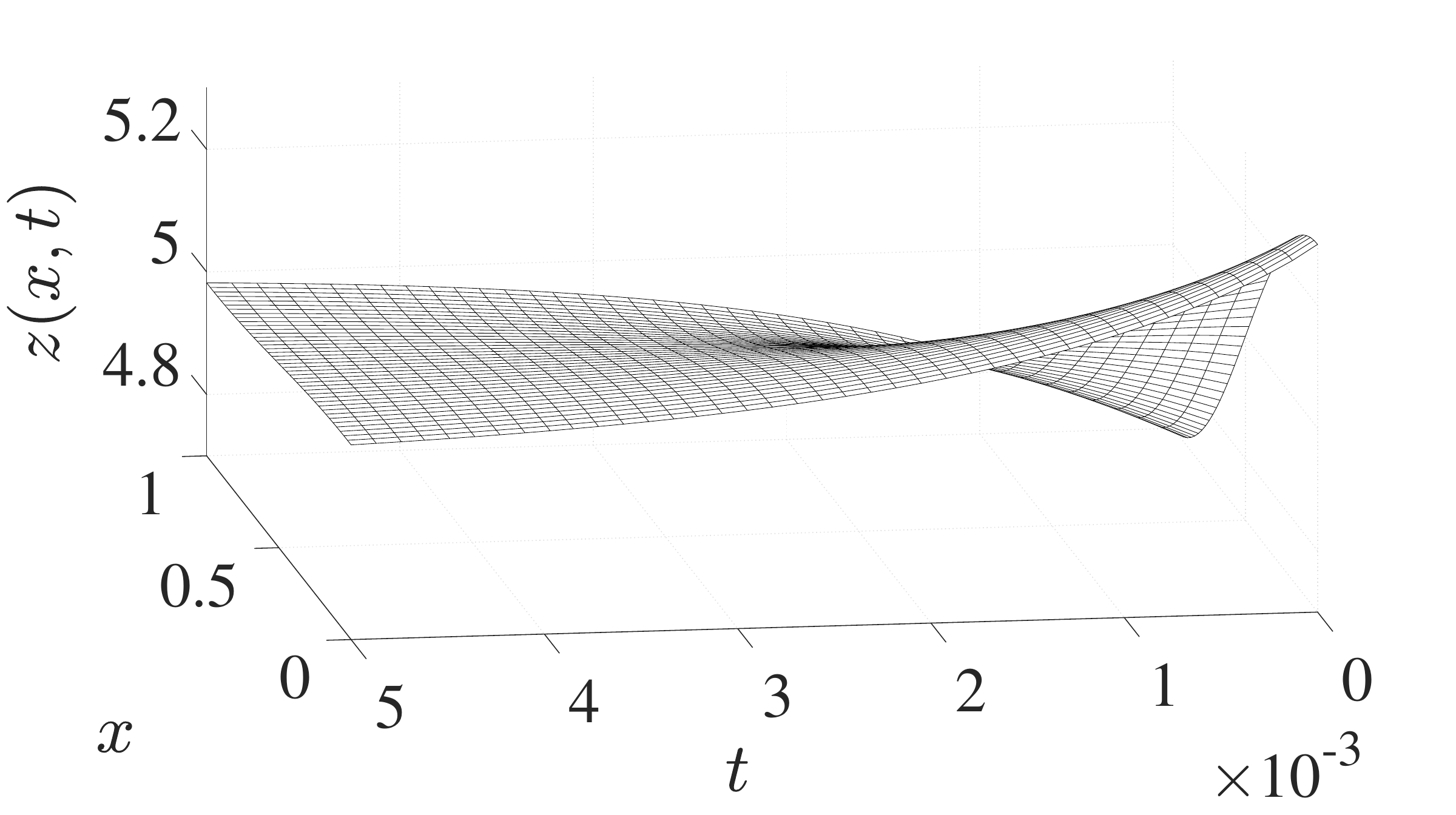}
    \caption{Response of $\Sigma$ (the state to be estimated).}
    \label{fig1}
\end{figure}

Based on the numerical simulation, we estimate $M=5.3$, $M' = 9.4 \times 10^{-1}$, and $M'' = 3.58$. To estimate $M'$ we used the Euler backward method for approximation of $z_x$, and to estimate $M''$ we used central differences for approximation of $z_{xx}$. Here, the constants $M$, $M'$, and $M''$ are tight because we have access to the simulation of $z$. In practice, however, we may only know that the state and its spatial derivatives cannot physically exceed certain large values. Consequently, the constants $M$, $M'$, and $M''$ will likely be larger, resulting in a more conservative estimate of the region of attraction.

We now aim to estimate the range of observer gains $\sigma$ and $a$ for which \eqref{cond_1}-\eqref{cond_2}-\eqref{cond_3} hold. To do so, we need upperbounds on the norms of the backstepping kernels appearing in \eqref{cond_1}-\eqref{cond_2}-\eqref{cond_3}. These bounds can be obtained by analyzing the series definitions of the kernels and their derivatives and deriving corresponding upperbounds. However, this may lead to relatively conservative estimates. Instead, we simulate the kernels over a range of observer gains and obtain the upperbounds numerically from these simulations. For the simulation of the kernels and their derivatives, we use their series representations and truncate the series after a finite number of terms, beyond which further terms have a negligible effect on the computed values. The set of admissible observer gains is illustrated in Figure \ref{fig2}.

\begin{figure}
    \centering
    \includegraphics[width=0.7\linewidth]{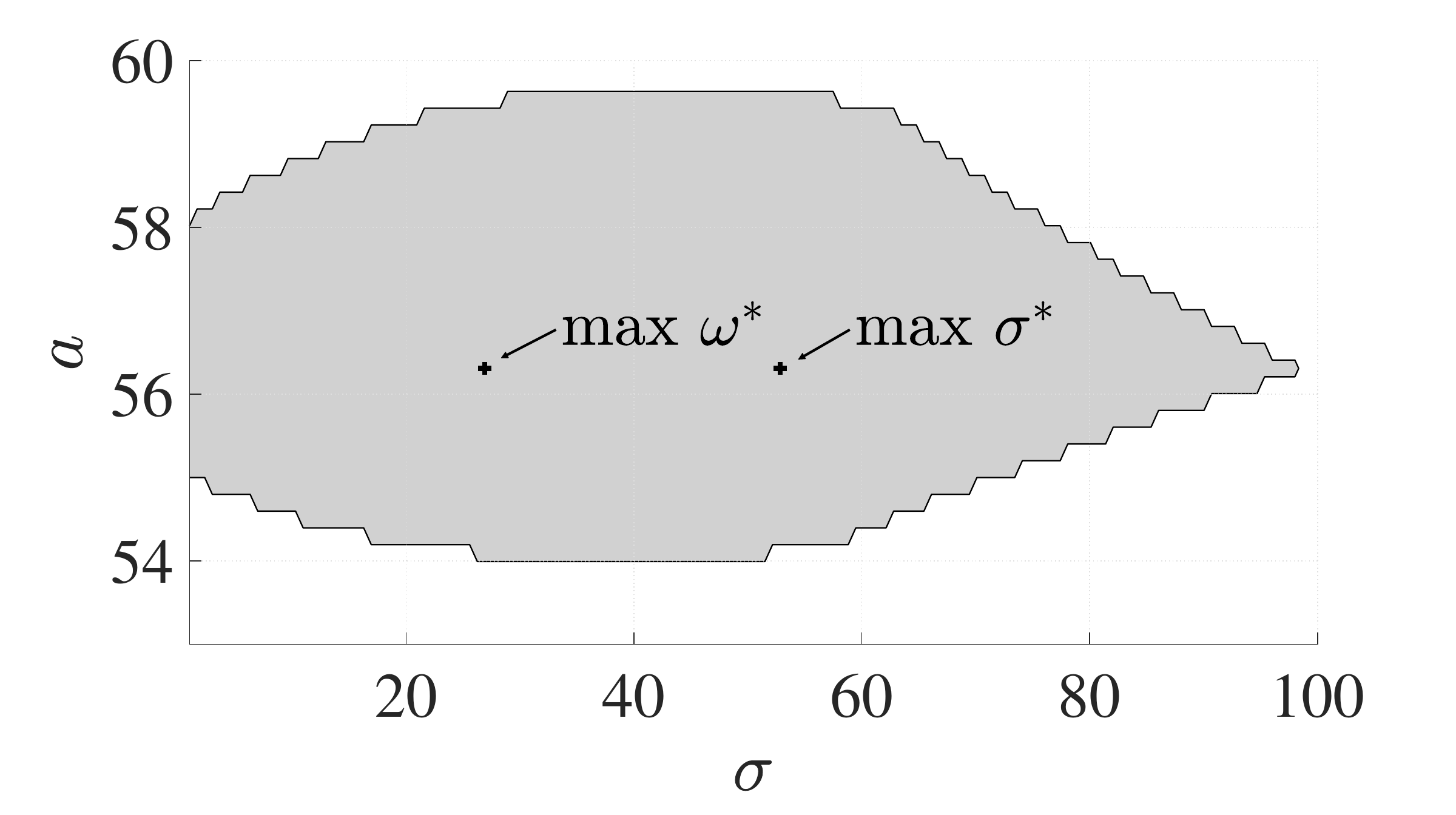}
    \caption{Set of observer gains for which \eqref{cond_1}-\eqref{cond_2}-\eqref{cond_3} are satisfied. Furthermore, $\max \omega^*$ corresponds to the couple $(\sigma,a)\approx (27,56)$ that maximizes the radius of the region of attraction, and $\max \sigma^*$ corresponds to the couple $(\sigma,a) \approx (53,56)$ that maximizes the convergence rate.}
    \label{fig2}
\end{figure}

Next, we estimate the radius $\omega^*$ of the region of attraction, satisfying the conditions \eqref{omega_cond_1}-\eqref{omega_cond_2}-\eqref{omega_cond_3}, as a function of $\sigma$ and $a$. Specifically, for each pair $(\sigma,a)$, we illustrate in Figure \ref{fig3} the largest radius of region of attraction satisfying \eqref{omega_cond_1}-\eqref{omega_cond_2}-\eqref{omega_cond_3}. Finally, Figure \ref{fig4} shows the estimated convergence rate $\sigma^*$ given by \eqref{conv_rate} as a function of $\sigma$ and $a$. It can be noted that the relationships between the various parameters are not monotonic, at least not based on our estimates.

\begin{figure}
    \centering
    \includegraphics[width=0.7\linewidth]{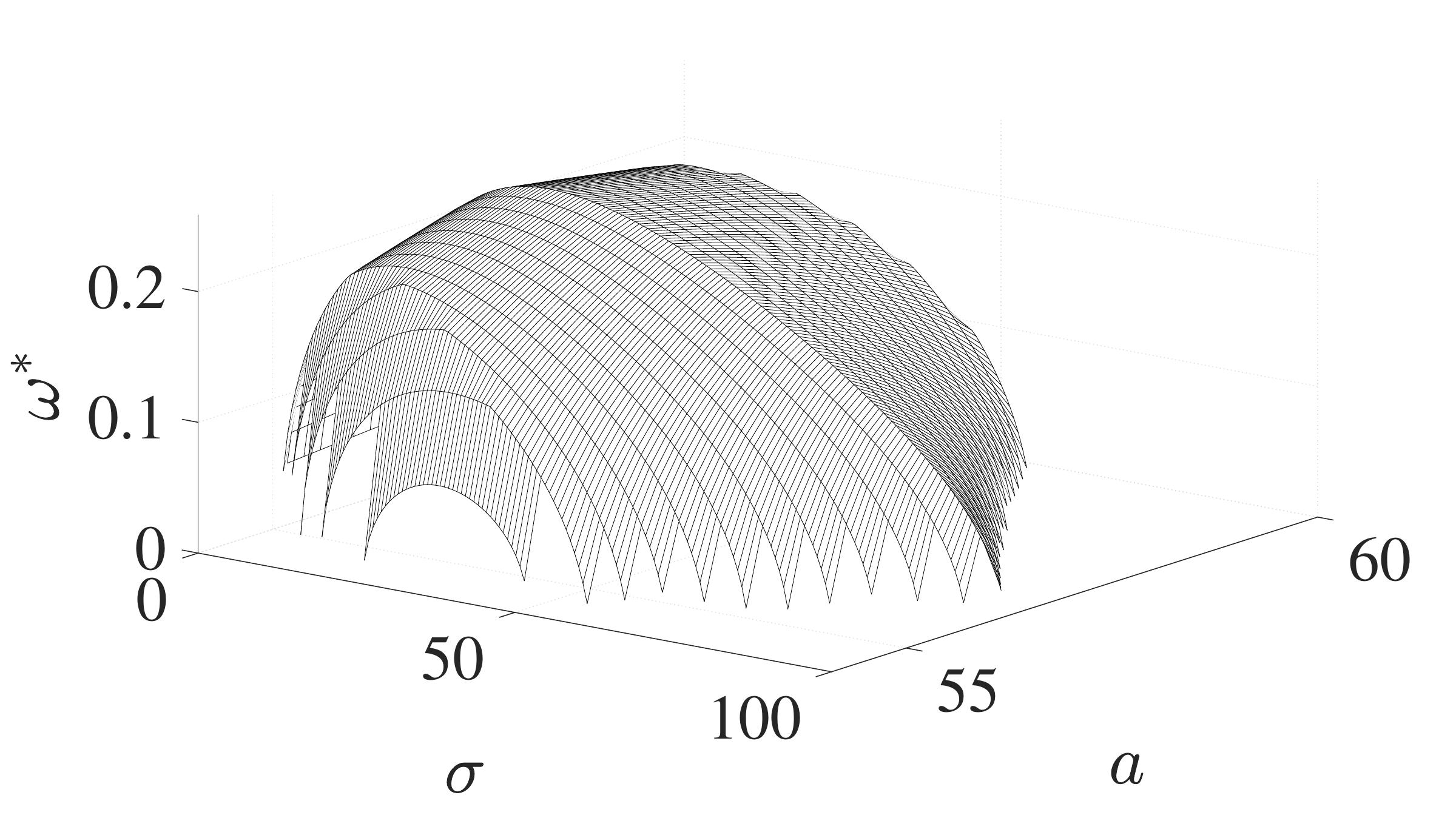}
    \caption{Radius $\omega^*$ of the region of attraction as a function of the observer gains $\sigma$ and $a$.}
    \label{fig3}
\end{figure}

\begin{figure}
    \centering
    \includegraphics[width=0.7\linewidth]{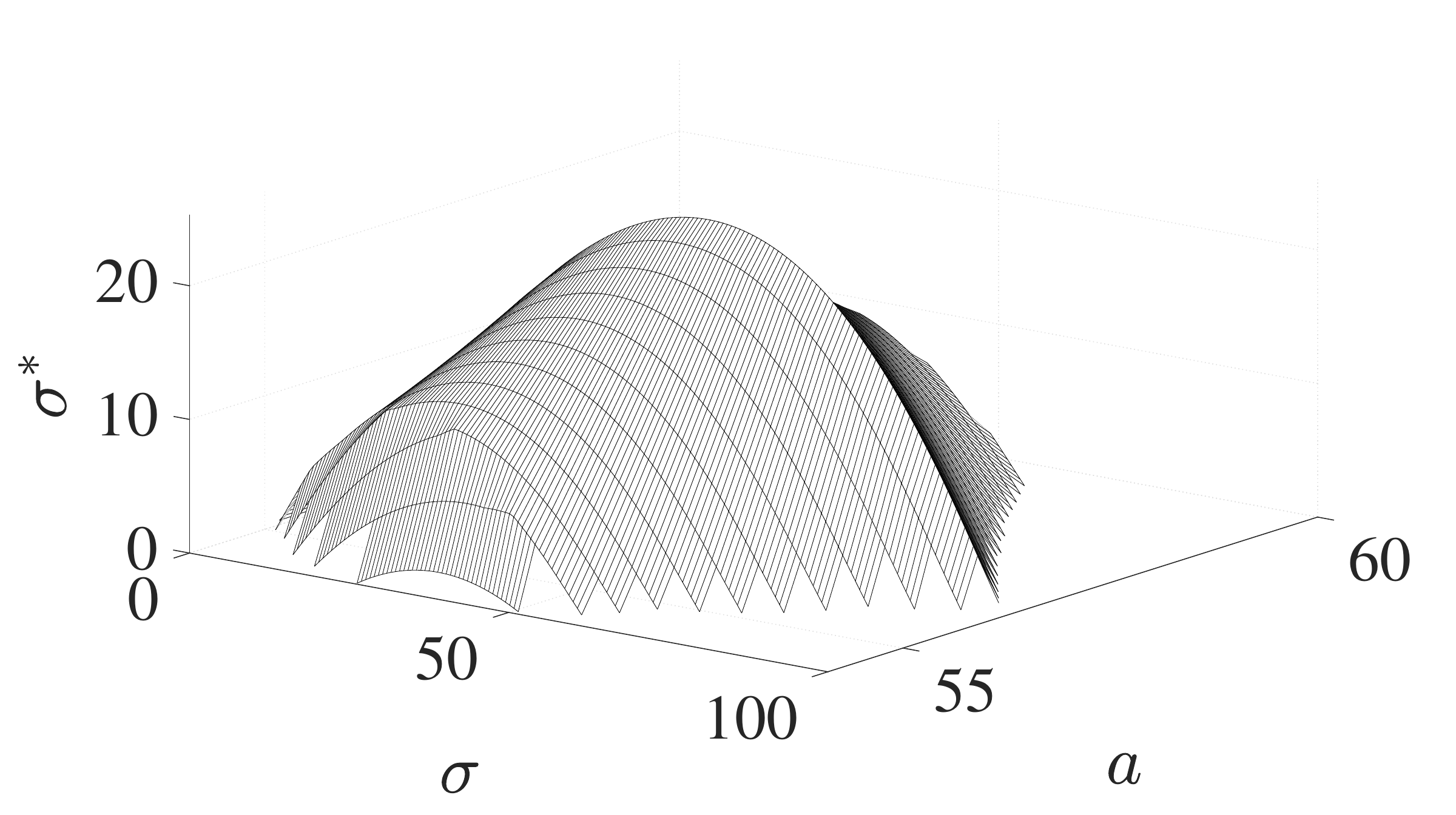}
    \caption{Convergence rate $\sigma^*$ as a function of $\sigma$ and $a$.}
    \label{fig4}
\end{figure}

Now, we consider the observer $\hat{\Sigma}_d$ for $\hat{h} = 1/20$ and $\widehat{\Delta t} = 5 \Delta t$, and $q=80$. The initial condition for the observer is taken as
\begin{align*}
\hat{z}_o(x) = z_o(x) - A\cos(\beta x) \quad \forall x\in [0,1],
\end{align*}
where $\beta$ is selected, for each pair of observer gains $(\sigma,a)$, to ensure the compatibility condition \eqref{compat2}, and $A \neq 0$ tunes $|e_o|_{H^1}$. Specifically, to guarantee \eqref{compat2}, we set $\beta \tan (\beta) = p_{10}$.

Let $(\sigma,\alpha) = (27,56)$, for which the estimated radius of the region of attraction is $\omega^* \approx 0.244$. In Figure \ref{fig5}, we simulate the $H^1$ norm of the estimation error $e=z-\hat{z}$ with the amplitude $A$ selected such that $|e_o|_{H^1} \in \{\omega^*, 10 \omega ^*, 100 \omega^*\}$. As can be seen, the estimation error still converges to zero even when the initial estimation error lies outside the estimated region of attraction, indicating that our estimates are conservative. We further increased the magnitude of the initial estimation error and observed that the error continued to converge to zero. We also performed several spatial and temporal mesh refinements for both $\Sigma$ and $\hat{\Sigma}$ to verify that the observed behavior remained unchanged. This raises the question of whether, for a given set of observer gains, the stability result is in fact global.

\begin{figure}
    \centering
    \includegraphics[width=0.7\linewidth]{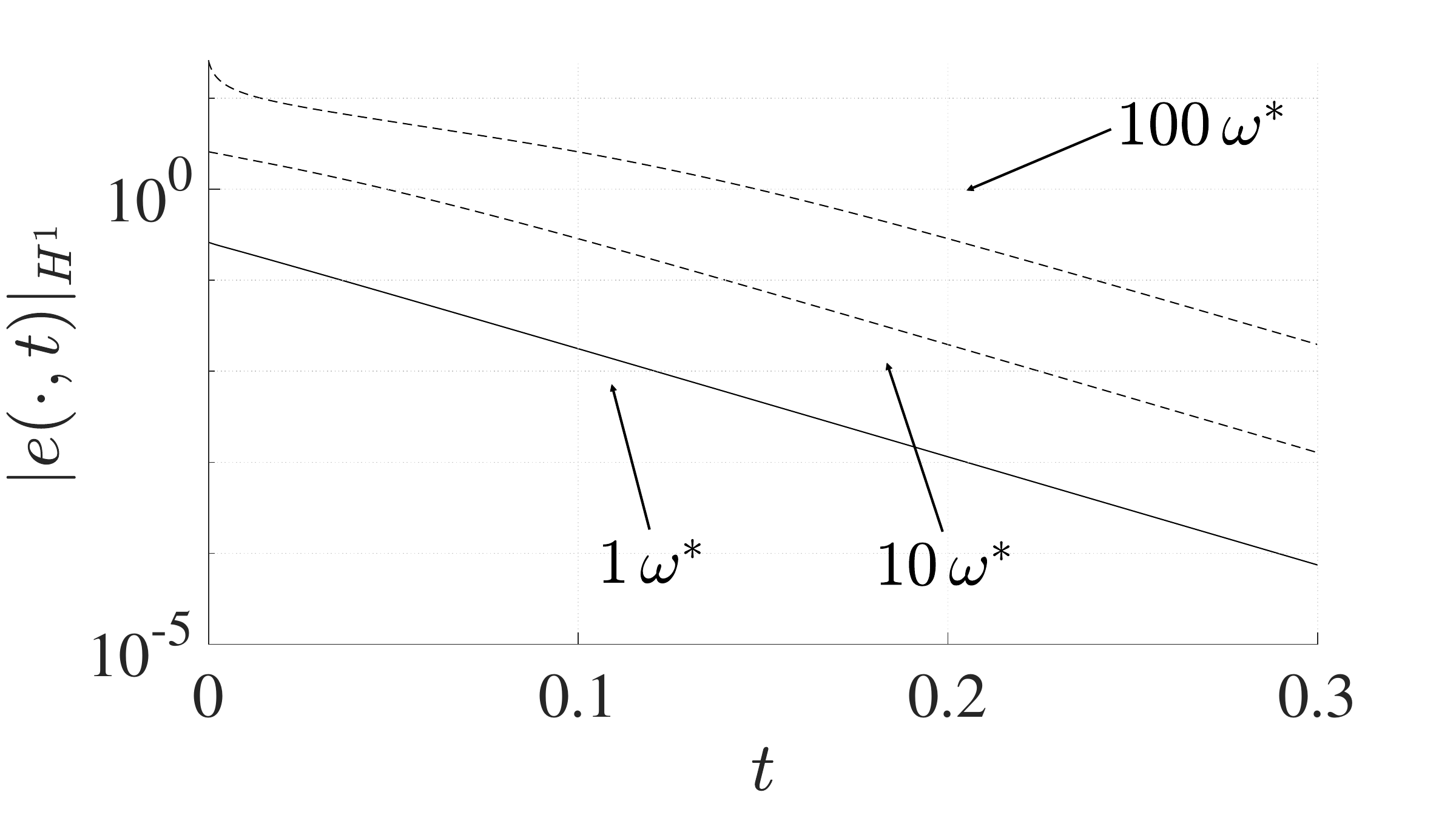}
    \caption{$H^1$ norm of the estimation error $e = z-\hat{z}$ for $|e_o|_{H^1}\in \{\omega^*, 10\omega^*, 100\omega^*
    \}$, with $\omega^*$ the estimated radius of the region of attraction for $(\sigma,a) = (27,56)$.}
    \label{fig5}
\end{figure}

Next, for a fixed size of the initial estimation error $|e_o|_{H^1} = 2.07\times 10^{-1}$ and the gain $\sigma =26.9$, we simulate in Figure \ref{fig6} the $H^1$ norm of the estimation error for $a\in \{5\times 10^{-2}, 56\}$. The value $a = 56$ lies within the admissible range, whereas $a= 5\times 10^{-2}$ does not. It can be seen that the estimation error does not converge to zero for $a= 5\times 10^{-2}$, which suggests that, although our estimates are conservative, they correctly capture the fact that $\bar{\delta}$ must not be too large. We repeated the simulations with several spatial and temporal mesh refinements to verify that this behavior remained unchanged.

\begin{figure}
    \centering
    \includegraphics[width=0.7\linewidth]{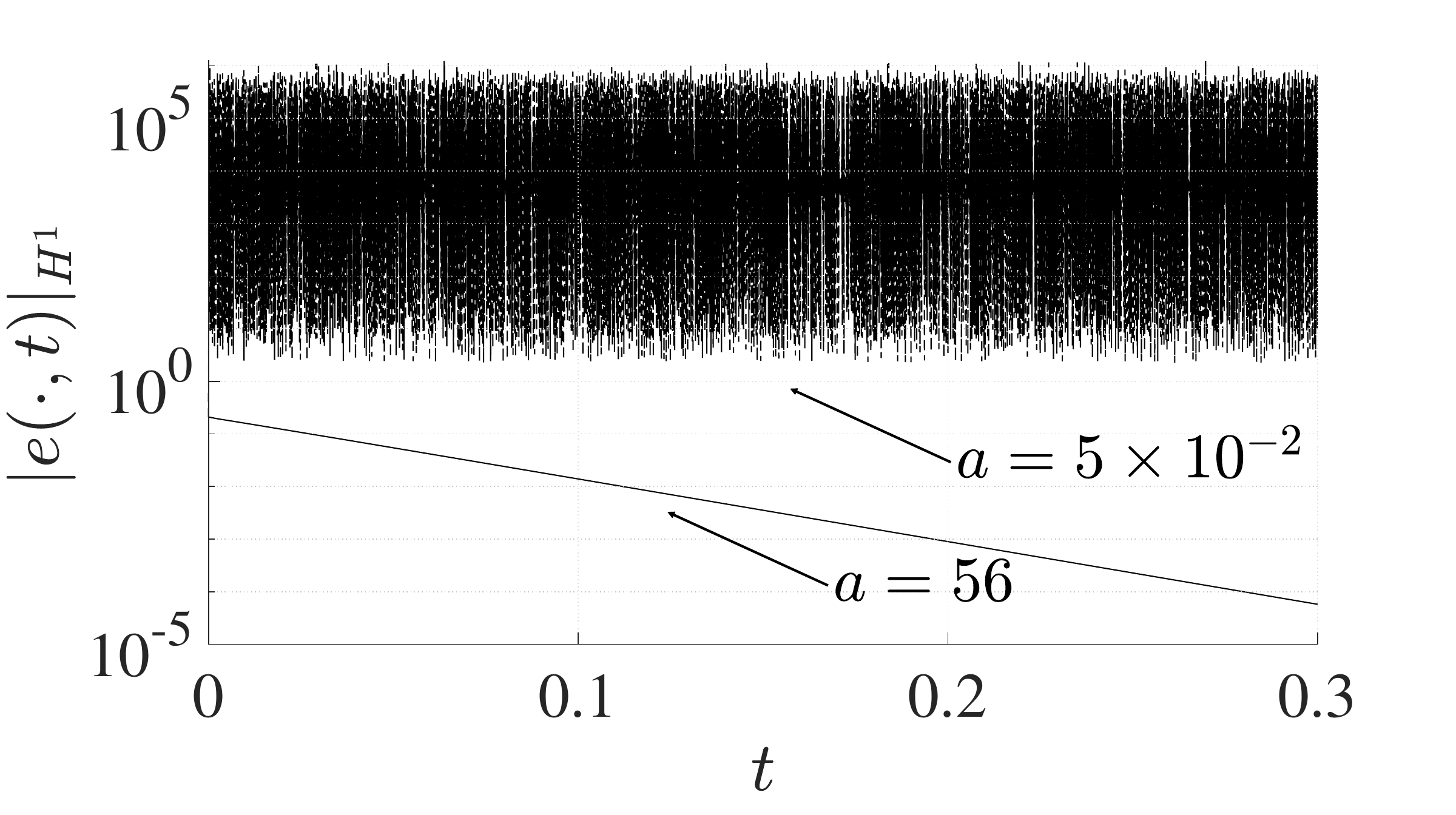}
    \caption{$H^1$ norm of the estimation error for $|e_o|_{H^1}=0.207$, $\sigma =26.9$, and $a\in \{5\times 10^{-2}, 56\}$.}
    \label{fig6}
\end{figure}

For the gain $\sigma$, the situation was less clear. For a fixed mesh $(h,\Delta t,\hat{h},\widehat{\Delta t})=(1/80,10^{-5},1/20,5\Delta t)$, we observed that increasing $\sigma$ to very large values (around $\sigma = 5\times 10^{3}$) caused the estimation error to no longer converge to zero, which is consistent with our Lyapunov analysis. However, after sufficiently refining the mesh, we were able to recover convergence of the estimation error to zero. This suggests that the apparent restriction on $\sigma$ may be a numerical artifact, and that the estimates obtained from our analysis may be very conservative.

Finally, we would like to assess the robustness of our observer to output noise. To this end, we replace $z_0^{kr}$ in the definition of $\hat{\Sigma}_d$ by $z_0^{kr}+v^{kr}$, where $v$ is Gaussian white noise generated using the MATLAB function \texttt{wgn}. Figure \ref{fig7} shows the $H^1$ norm of the estimation error for zero noise and for noise levels up to $10\%$ of $|e_o|_{H^1} = 0.207$, and for the observer gains $(\sigma,a) = (26.9,56.3)$. The simulations suggest that an ISS result with respect to the output noise may be possible, although estimating a region of attraction towards a ball whose radius is proportional to the noise bound does not appear to be straightforward.

\begin{figure}
    \centering
    \includegraphics[width=0.7\linewidth]{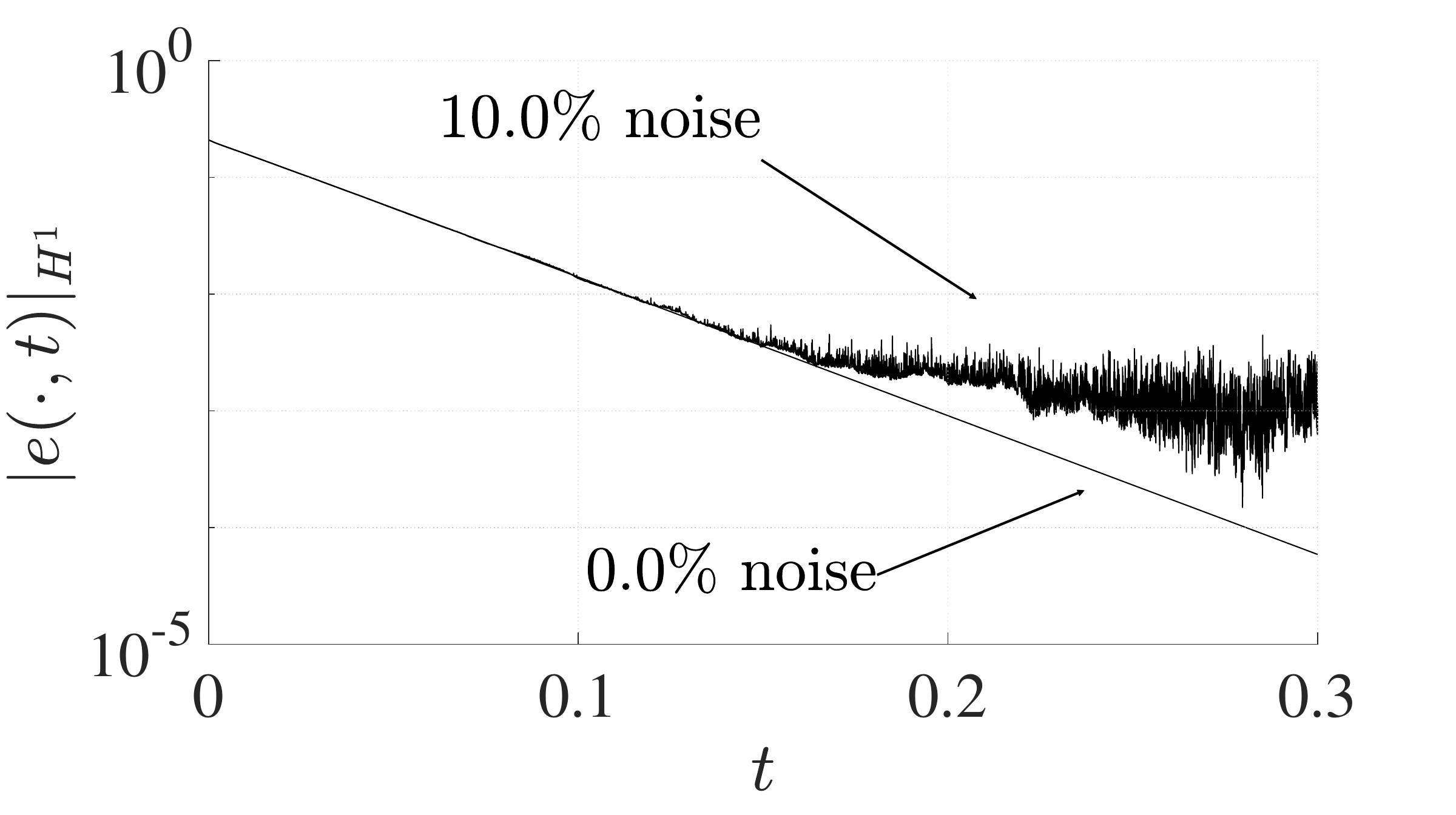}
    \caption{$H^1$ norm of the estimation error without output noise vs. with output noise whose power is up to $10\%$ of $|e_o|_{H^1}$.}
    \label{fig7}
\end{figure}

\section{Conclusion and Research Perspectives}\label{sec_conclusion}
In this paper, we designed a nonlinear backstepping observer for the quasilinear heat equation. The observer employs correction gains computed for a linear heat equation with constant diffusivity. We established exponential stability of the zero equilibrium for the estimation error dynamics in the $H^1$ norm, together with estimates of the region of attraction and the convergence rate. Our results extend to the case where $\Sigma$ is affected by external known inputs, provided that they are sufficiently regular for a forward-complete classical solution $z$ to exist and that Assumption \ref{ass2} holds. The extension is straightforward, as these inputs can simply be incorporated into the observer so that they disappear from the dynamics of the estimation error. The case where these inputs are unknown, or where the output is corrupted by unknown noise, is less straightforward, however, because our Lyapunov analysis relies on a Lyapunov functional whose time derivative is nonlinear, rather than affine, in the functional itself, making the estimation of a region of attraction for a local ISS-type result difficult to obtain, to the best of our knowledge.

Several other directions for future research remain. For example, it would be of interest to extend our approach to two- and three-dimensional cases. This would require, however, establishing stability of the zero equilibrium for the estimation error dynamics in the $H^2$ norm rather than the $H^1$ norm. Moreover, as observed in the simulations in Section \ref{example}, our estimates are likely to be quite conservative. In particular, although our analysis correctly captures the fact that the gain $a$ cannot be chosen arbitrarily, the estimation error was observed to converge to zero regardless of the size of the initial estimation error. This raises the question of whether global exponential stability can be established. 

\section{Acknowledgement}
The support of a grant from National Science and Engineering Research Council of Canada for this research is gratefully acknowledged.



\appendix

\par\noindent 
\parbox[t]{\linewidth}{
\noindent\parpic{\includegraphics[height=1.5in,width=1in,clip,keepaspectratio]{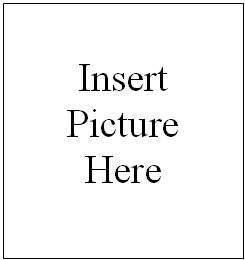}}
\noindent {\bf Mohamed Camil Belhadjoudja}\ received the Ph.D. degree in automatic control from Grenoble Alpes University, France, in 2025, under the supervision of M. Maghenem and E. Witrant. From January to August 2026, he was a Postdoctoral Fellow at the University of Waterloo, Canada, under the supervision of K. A. Morris. Since September 2026, he has been an Associate Professor with Grenoble Alpes University, GIPSA-LAB, France. His research interests include the control and estimation of infinite-dimensional systems.
}

\vspace{1cm}

\par\noindent 
\parbox[t]{\linewidth}{
\noindent\parpic{\includegraphics[height=2in,width=1in,clip,keepaspectratio]{insert-picture-here.jpg}}
\noindent {\bf Kirsten A. Morris}\ (Fellow, IEEE) received the
Ph.D. degree in electrical engineering with
thesis “Finite-dimensional control of infinite-
dimensional systems” from the Faculty of Engineering, University of Waterloo, Waterloo, ON,
Canada, in 1989.
She has authored or coauthored several
books Controller Design for Distributed Parameter Systems, and Introduction to Feedback Control, and was the editor of the book Control of
Flexible Structures. She has held visiting positions with NASA Langley (USA), the Fields Institute (Canada), Institute
for Mathematics and Applications (USA), and the Institut de Mathematiques de Bordeaux (France). Her research interests include control and
estimation of distributed parameter systems and smart materials, and
improving performance through attention to actuator/sensor location.
Ms. Morris is a Fellow of SIAM and IFAC. She was a Vice-President
of the IEEE Control System Society from 2013 to 2016, Vicechair of the
SIAM Control and Systems Theory group from 2016 to 2017 and Chair
from 2018 to 2019.}

\end{document}